\documentclass[acmsmall,screen]{acmart}

\startPage{1}

\usepackage{bbm}
\usepackage{float}
\usepackage{enumitem}
\usepackage{color}
\usepackage{multicol}

\newcommand{\BEQ}{::=}
\newcommand{\BOR}{\mid}
\newcommand{\bool}{\mathtt{bool} }
\newcommand{\real}{\mathtt{real} }
\newcommand{\unit}{\mathtt{unit} }
\newcommand{\op}{\mathrm{op} }
\newcommand{\pred}{\mathrm{pred}}

\newcommand{\myif}{\mathtt{if}\;}
\newcommand{\mythen}{\;\mathtt{then}\;}
\newcommand{\myelse}{\;\mathtt{else}\;}

\renewcommand{\G}{\Gamma}
\newcommand{\D}{\Delta}

\newcommand{\dRD}[3] {#3.\mathtt{rd}_{#2}(#1)}
\newcommand{\fRD}[3] {#3.\mathtt{fd}_{#2}(#1)}

\newcommand{\rop}{\op_r}

\newcommand{\newcalR}[4]{#4.\mathcal{R}_{#3} (#1.\,#2 )}
\newcommand{\newcalRprop}[5]{#4.\mathcal{R}^{#5}_{#3}(#1.\,#2)}
\newcommand{\mylet}{\mathtt{let}\; }
\newcommand{\mybe}{\; = \; }
\newcommand{\myin}{\;\mathtt{in}\; }

\newcommand{\myletrec}{\;\mathtt{letrec}\; }
\newcommand{\amyletrec}[1]{\;\mathtt{letrec}_n }

\newcommand{\Dom}{\mathrm{Dom}}
\newcommand{\FV}{\mathrm{FV}}
\newcommand{\FFV}{\mathrm{FFV}}
\newcommand{\imp}{\Rightarrow} 

\newcommand{\impe}{\Rightarrow} 

\newcommand{\imps}{\leadsto}  
\newcommand{\den}[1]{[\![ #1 ]\!]}
\newcommand{\true}{\mathtt{true}}
\newcommand{\false}{\mathtt{false}}

\newcommand{\tuple}[1]{\langle #1 \rangle}

\newcommand{\mycut}[1]{}

\newcommand{\R}{\mathbb{R}}
  \newcommand{\grad}[2]{\mathtt{grad}_{#1}(#2)}

 \newcommand{\type }{\!:\!}
  
  \newcommand{\fclo}[2]{\mathbf{clo}_{#1} (#2)}
  \newcommand{\afclo}[3]{\mathbf{clo}_{#1,#2} (#3)}

\newcommand{\POPLomit}[1]{}

\newcommand{\POPLsomit}[1]{}

\newcommand{\T}{\mathbb{T}}
\newcommand{\dtrue}{t\!t}
\newcommand{\dfalse}{f\!\!\hspace{-1pt}f}
\newcommand{\pr}{\rightharpoonup}
\newcommand{\res}{\upharpoonright}
\newcommand{\eqdef}{=_{\small \mbox{def}}}

\newcommand{\ttfst}{\mathtt{fst}}
\newcommand{\ttsnd}{\mathtt{snd}}

                   \newcommand{\ov}[1]{\overline{#1}}
       
       \newcommand{\bx}{\mathbf{x}}
       \newcommand{\by}{\mathbf{y}}
       
       \newcommand{\bu}{\mathbf{u}}
\newcommand{\bv}{\mathbf{v}}

       \newcommand{\Cl}{\mathrm{Cl}}

\newcommand{\Mat}{\mathrm{Mat}}

\newcommand{\Op}{\mathrm{Op}}
\newcommand{\Pred}{\mathrm{Pred}}

\newcommand{\ev}{\mathrm{ev}}
\newcommand{\bev}{\mathrm{bev}}
 \newcommand{\Val}{\mathrm{Val}}

 \newcommand{\hole}{[\;]}
\newcommand{\jacobian}{\mathrm{J}}
\newcommand{\myd}{\mathrm{d}}

\newcommand{\N}{\mathbb{N}}

\newcommand{\mydr}{\mathrm{d}^r}

\newcommand{\smooth}[2]{\mathcal{S}[#1,#2]}

\newcommand{\cont}[2]{\mathcal{C}[#1,#2]}

\newcommand{\fst}{\pi_0}
\newcommand{\snd}{\pi_1}

\newcommand{\myundef}{\uparrow}

\newcommand{\mydef}{\downarrow}
\newcommand{\myDef}{\Downarrow}
\newcommand{\ldash}{\vdash_l}

\newtheorem{fact}{Fact} 

\newcommand{\mycomment}[1]{}
\newcommand{\ap}[1]{#1}

\newcommand{\mylcomment}[1]{}

\newcommand{\myacomment}[1]{}

\newcommand{\dtimes}{\mathbin{\dot{\times}}}
\newcommand{\dlessthan}{\mathbin{\dot{<}}}
\newcommand{\plainrdstring}{\mathtt{rd}}
\newcommand{\myletop}{\mathtt{let}}
\newcommand{\myletrecop}{\mathtt{letrec}}
\newcommand{\myletrecnospacebefore}{\mathtt{letrec}\; }
\newcommand{\myinnospacebefore}{\mathtt{in}\; }

\newcommand{\myParameterizedModel}{\mathit{myUntrainedModel}}
\newcommand{\myLoss}{\mathit{myLoss}}
\newcommand{\rate}{\mathit{rate}}

\newcommand{\myTrainedModel}{\mathit{myTrainedModel}}
\newcommand{\myLossNow}{\mathit{currentLoss}}
\newcommand{\maxLoss}{\mathit{maxLoss}}
\newcommand{\gradLossNow}{\mathit{gradLoss}}
\newcommand{\descend}{\mathit{descend}}

\newcommand{\s}{\hspace{1pt}}
\setcopyright{rightsretained}
\acmPrice{}
\acmDOI{10.1145/3371106}
\acmYear{2020}
\copyrightyear{2020}
\acmJournal{PACMPL}
\acmVolume{4}
\acmNumber{POPL}
\acmMonth{1}

\begin{document}

\title
{A Simple Differentiable Programming Language}   
\titlenote{This version of the paper is the POPL publication version, but with some minor corrections.}   


\author{Mart\'{i}n Abadi}
\affiliation{
  \institution{Google Research}            
  \country{United States}                    
}
\email{abadi@google.com }        

\author{Gordon D.\ Plotkin}
\affiliation{
  \institution{Google Research}           
  \country{United States}                   
}
\email{plotkin@google.com}         

\begin{abstract}
Automatic differentiation plays a prominent role in scientific computing and in modern machine learning, often in the context of powerful programming systems. The relation of the various embodiments of automatic differentiation to the mathematical notion of derivative
  is not always entirely clear---discrepancies can arise, sometimes inadvertently.
In order to study automatic differentiation in such programming contexts,
we define a small but expressive programming language that includes 
a construct for reverse-mode differentiation. 
We give operational and denotational semantics for this language. The operational semantics employs 
popular implementation techniques, while the denotational semantics employs 
notions of differentiation familiar from real analysis.
We establish that these semantics coincide.
\end{abstract}



\maketitle




\POPLomit{
\tableofcontents}

\section{Introduction}\label{intro}

Automatic differentiation is a set of techniques for calculating the derivatives of functions described by computer programs  (e.g.,~\cite{DBLP:journals/toplas/PearlmutterS08,TapenadeRef13,JMLR:v18:17-468,Gri00}). These techniques are not required to produce symbolic representations for derivatives 
as  in classic symbolic differentiation; on the other hand, neither do they employ finite-difference approximation methods common in numerical 
differentiation. Instead, they rely on the chain rule from calculus to obtain the desired derivatives from those of the programs's basic operations. Thus, 
automatic differentiation is at the intersection of calculus and programming. However, the programs of interest are more than chains of operations: they 
may include control-flow constructs, data structures,  and computational effects (e.g., side-effects or exceptions). Calculus does not provide an immediate 
justification for the treatment of such programming-language features.

In the present work we 
help bridge the gap between rules for automatic differentiation in expressive programming languages
 and their mathematical justification in terms of denotational semantics. Specifically,
we consider automatic differentiation from a programming-language
perspective by defining and studying a small but powerful functional first-order language.
The language has 
conditionals and recursively defined functions (from which loops can
be constructed), but only rudimentary data structures. Additionally, it contains a construct for reverse-mode
differentiation, explained in detail below. Our
language is thus inspired by modern systems for machine learning,
which include standard programming constructs and support reverse-mode
differentiation.
 Reverse-mode differentiation permits the
computation of gradients,  forward-mode derivatives, and more.  Indeed as our differentiation construct  is a language primitive,  differentiations can be nested within differentiations, allowing the computation of higher-order derivatives.





In the setting of a language such as ours, we can consider some common approaches to implementing differentiation: 
\begin{itemize} 
\item 

One approach relies on code transformation, whether on source code or intermediate representations. 
For example, for the derivative of a conditional expression 
$\myif B \mythen M_1 \myelse M_2$, it would output $\myif B \mythen N_1 \myelse N_2$, where $N_1$ and $N_2$ are the derivatives of $M_1$ and $M_2$ respectively.
This approach is employed, for instance, in Theano~\cite{bergstra2010theano}, TensorFlow 1.0~\cite{tensorflow2016,Yu:2018:DCF:3190508.3190551}, and Tangent~\cite{tangent}.

\item Another approach relies on tracing, typically eliminating control  structures to produce a simpler form of code,  
which we call an \emph{execution trace}, that can more easily be differentiated. For example,  to produce the derivative of $\myif B \mythen M_1 \myelse M_2$, tracing would evaluate the conditional and produce a trace of the branch taken.
Execution traces correspond to graphs of basic operations, and can be taken to be sequences of elementary assignments or else functional programs in 
A-normal form. Their derivatives can be calculated by applying the chain rule to those basic operations, perhaps via a code transformation (but now of a much simpler kind).
Tracing  may also record some intermediate values in an \emph{evaluation trace},  
to reduce, or eliminate, the need for recomputation.\footnote{Terminology in the automatic differentiation literature varies. Here we follow~\cite{Gri00} for evaluation traces.  Our execution traces are, perhaps in somewhat different  manifestations, variously termed Wengert lists or  tapes or evaluation traces~\cite{JMLR:v18:17-468,DBLP:journals/toplas/PearlmutterS08}. They can also be seen as combinations of   the operation and index traces  of~\cite{Gri00}.}


This approach thereby conveniently avoids the problem of defining code transformations for conditionals and many other language constructs.
It can also be implemented efficiently, sometimes in part with JIT compilation.
For these reasons, trace-based differentiation is of growing importance. 
It is employed, for instance, in Autograd~\cite{autograd}, TensorFlow Eager Mode~\cite{agrawal2019tensorflow}, Chainer~\cite{chainerorg}, PyTorch~\cite{pytorch}, and JAX~\cite{JAX}\footnote{See \url{ https://www.sysml.cc/doc/146.pdf}.}.
\end{itemize}

We therefore focus on trace-based differentiation, 
and give our language an operational semantics using the trace-based approach.  To do so, we define a sublanguage of execution trace terms (called simply trace terms below). These 
have no conditionals, function definitions or calls, or reverse-mode differentiations. They do have local definitions, corresponding to fanout in the graphs, but may not be  in A-normal form. Tracing  
is modeled by a  new kind  of  evaluation,  called symbolic evaluation. This uses an environment for the free variables of a term to remove conditionals and function calls. 
Function derivatives at a given value are evaluated  in three stages: first, the function is traced at that value; next, the resulting  trace term is symbolically differentiated (largely just using the chain rule), resulting in another such trace term; and,  finally, that term is evaluated.

We do not account for some of the optimizations used in practice.
Doing so would have been a more complicated enterprise, possibly with more arbitrary choices tied to implementation details, and we wished to get a more straightforward formalization working first.



From a mathematical perspective, both approaches to implementing differentiation pose correctness problems.
In particular, functions defined using conditionals need not
be continuous, let alone differentiable. Consider, for example, the following definition
%
\[f(x\type\real)\type \real \mybe \myif (x < 0) \mythen 0 \myelse x\]
of the popular ReLU function \cite{goodfellow2016book}. This function is not differentiable at $0$. Further, changing the function body to $\myif (x < 0) \mythen 0 \myelse 1$ yields a non-continuous function.
What is more, both approaches can produce wrong answers even for differentiable functions! 
Consider, for example, the following definition of the identity function on the reals: 
%
\[g(x\type\real)\type \real \mybe \myif (x = 0) \mythen 0 \myelse x\]
The 
 derivative of this function at $x=0$ is $1$. However, differentiation ``by branches'' (whether by code transformation or tracing)
would produce the wrong answer,~$0$.

In order to capture the mathematical perspective,  in addition to its operational semantics we give our language a denotational
semantics.  This
semantics is based on classical notions of differentiation
from real analysis (see, for example,~\cite{Tr13}). That theory concerns multivariate functions on the reals defined on open domains, i.e., partial such functions with open domains of definition. In our semantics, we make use of those that are smooth  (that is, those that can be partially differentiated any number of times). 
A particularly pleasing aspect of this mathematical development  is how well domain theory (needed to account for recursion) interacts with differentiation.



Partiality is necessary, as for 
any language with general recursion, but it also gives us useful flexibility in connection with differentiation.
For example,  let $\dlessthan$ be 
the approximation to $<$ which is equal to it except on the diagonal (i.e., where both arguments are equal) where it is undefined. Then
\[\dot{f}(x\type\real)\type \real \mybe \myif (x \dlessthan 0) \mythen 0 \myelse x\]
defines an approximation to ReLU which is undefined at $0$. The approximation to $<$ is (unlike $<$) continuous
(i.e., the pre-images of $\true$ and $\false$ are open sets), and the approximation to ReLU is differentiable wherever it is defined. 
Therefore, we design the semantics of our language so that it forbids functions such as $f$ but allows related 
approximations such as $\dot{f}$. An interesting question is how satisfactory an idealization this is of programming practice (which in any case works with approximate reals). We return to this point in the final section.


%


Proceeding in this way, we obtain adequacy theorems (i.e., operational soundness and completeness theorems) 
connecting the operational semantics of our language with a denotational semantics based on the classical theory of differentiation of partially defined multivariate real functions. Our theorems apply not only to conditional expressions but to the full language.



In sum, the main contributions of this paper are: 
(1) a first-order language with conditionals, recursive function definitions, and a  reverse-mode differentiation construct; 
(2) an operational semantics that models one form of trace-based differentiation;
(3)  a denotational semantics based on standard mathematical notions from real analysis and domain theory;
and (4) theorems that show that the two semantics coincide, i.e.,  the
derivatives computed by the operational semantics are indeed the
correct derivatives in a mathematical sense. Beyond the specifics of
these results, this paper aims to give some evidence of the relevance of 
ideas and techniques from the programming-languages literature
for programming systems that
include automatic differentiation, such as current systems for machine learning.

\subsubsection*{Additional context}

While traditionally associated with scientific computing, automatic differentiation is now a central component of many modern machine learning  systems, and  those for deep learning in particular  \cite{goodfellow2016book}. These systems often employ automatic differentiation to compute the gradients of ``loss functions'' with respect to parameters, such as neural network weights. Loss functions measure the error resulting from particular values for the parameters. For example, when a machine-learning model is trained with a dataset consisting of pairs $(x_0,y_0), \ldots,(x_{n-1},y_{n-1})$, aiming to learn a function that maps the $x_i$'s to the $y_i$'s, the  loss function may be the distance between the $y_i$'s and the values  the model predicts when presented with the $x_i$'s. By applying gradient descent to adjust the parameters, this error can be reduced, until convergence or (more commonly) until the error is tolerable enough. This simple approach has proven remarkably effective: it is at the core of many recent successes of machine learning in a variety of domains.


Whereas gradients are for functions of type
$\real^n \rightarrow \real$, for $n \geq 0$, treating the more general
functions of type $\real^n \rightarrow \real^m$, for $n,m \geq 0$, works better with function composition,
and with the composite structures such as tensors of reals used in deep learning. 
The
literature contains two basic ``modes'' for differentiating such functions. 
Forward-mode 
extends the computation
of the function, step by step, with the computation of derivatives; it
can be seen as evaluating the function on dual numbers of the form $v
+ \dot{v}\varepsilon$ where $\varepsilon$ is nilpotent.
 In contrast,
reverse-mode 
propagates derivatives backwards from
each output, typically after the computation of the
function. Reverse-mode differentiation is often preferred because of
its superior efficiency for functions of type
$\real^n \rightarrow \real^m$ with $m <\!\!< n$. In particular,
systems for machine learning,
which often deal with loss functions for which $m = 1$,
generally rely on reverse-mode differentiation. We refer the reader to 
the  useful recent survey~\cite{JMLR:v18:17-468} for additional background on these two modes of differentiation; it also discusses the use of higher-order differentiation.

Applications to machine learning are our main motivation. 
Accordingly, our language is loosely inspired by systems for machine learning, and the implementation strategies that we consider are ones of current interest there. 
We also de-emphasize some concerns (e.g., numerical stability) that, at present, seem to play a larger role in scientific computing than in machine learning. As 
noted in~\cite{DBLP:journals/corr/BaydinPS16a}, the machine learning community has developed a 
mindset and a 
body of techniques distinct from those traditional in automatic differentiation.


The literature on scientific computing has addressed the correctness problem for conditionals~\cite{Beck:1994:IAD:195983.196017,Fischer2001ADR}, although not in the context of a formally defined programming language. 
In~\cite{May02} a formal proof of correctness 
for an algorithm for the automatic differentiation of straight-line sequences of Fortran assignments was given using the Coq   theorem prover \cite{Coq}. 
Closer to machine learning,~\cite{SLD17} consider a stochastic graphical formalism where the nodes are random variables, and use the Lean theorem prover~\cite{Lean} to establish the correctness of stochastic backpropagation. 
However, overall, the literature does
not seem to contain semantics and theorems for a language of the kind we consider here.
%


%
%

Our work is also related to important papers  by Ehrhard, Regnier, et al.~\cite{ehrhard2003differential},  and by Di Gianantonio and Edalat~\cite{GE13}.  
Ehrhard and Regnier introduce the differential $\lambda$-calculus; 
this
 is a simply-typed higher-order $\lambda$-calculus with a forward-mode differentiation construct which can be applied to functions of any type. 
It can be modeled using the category of convenient vector spaces and smooth functions between them (see~\cite{blute2010convenient,kriegl1997convenient}). 
Ehrhard and Regnier do not give an operational semantics but they do give rules for symbolic differentiation and it should not be too difficult to use them to give an operational semantics. However their language with its convenient vector space semantics only supports total functions. It therefore 
cannot be extended to include recursive function definitions or conditionals (even with total predicates, as continuous functions from $\R^n$ to the booleans are constant).
Di Gianantonio and Edalat  prove adequacy theorems for their language, as do we, but their work differs from ours in several respects. In particular,  their language has first-order forward-mode but no reverse-mode differentiation: our language effectively supports both, and 
at all orders. On the other hand,  their language allows recursively-defined higher-order functions and accommodates functions, such as the ReLU function, which are differentiable in only a weaker sense. 
As far as we know, no other work 
on differentiable programming languages (e.g.,~\cite{DBLP:journals/toplas/PearlmutterS08,Elliott-2018-ad-icfp,DBLP:journals/corr/abs-1803-10228,DBLP:journals/corr/abs-1806-02136,Man12})  gives operational and denotational semantics and proves adequacy theorems.
Further afield, there is a good deal of work in the categorical literature on categories equipped with differential structure, for example~\cite{blute2009cartesian,bucciarelli2010categorical}.
 
\subsubsection*{Outline}

Section~\ref{language} defines our language.
Section~\ref{opsem} gives it an operational semantics with rules for symbolically evaluating general terms to   trace terms, and for symbolically differentiating these terms. 
%
%
Sections~\ref{math} and~\ref{densem} cover the needed mathematical material and the denotational semantics.
Sections~\ref{adequacy} establishes the correspondence between operational and denotational semantics.
Section~\ref{query} concludes with discussion and some suggestions for future work. 

\section{A simple language} \label{language}














The \emph{types} $S, T, U,  \ldots$  of our language are given by the grammar:
%
%
\[\vspace{0pt}\begin{array}{lcl}T & \BEQ  & \real 
                                \BOR  \unit  \BOR T \times U 
             \end{array}\]
%
We will make use of iterated products $T_0 \times \ldots \times T_{n-1}$, defined to be $\unit$ when $n= 0$, $T_0$ when $n =1$, and, recursively,  $(T_0 \times \ldots \times T_{n-1}) \times T_n$, when $n > 1$; we write $\real^n$ for the $n$-fold iterated product of $\real$. Note that this type system includes the types of tensors (multidimensional arrays) of a given shape: the type of tensors of shape $(d_0,\ldots,d_{n-1})$ is the iterated product  $\real^{d_0} \times \ldots \times \real^{d_{n-1}}$. 
%
%
The  \emph{terms} $L,M,N,P, \ldots$ and  \emph{boolean terms} $B$  of the language are built from 
\emph{operation symbols} $\op \in \Op$ and \emph{predicate symbols} $\pred \in \Pred$. An example operation symbol could be 
$ \mathit{DProd}_{n}$ for dot product of vectors of dimension $n$  (for $n \in \mathbb{N}$); an example predicate symbol could be $\dlessthan$.

The terms are given by the following grammar, 
 where  $x$ and $f$  range over disjoint countably infinite alphabets of  ordinary  and function variables, respectively.
  We assume available a standard ordering of the function variables.
%
\[\hspace{-20pt}\begin{array}{lcl}
M & \BEQ  & x  \BOR   r \;\; (r\in \R) \BOR M + N  \BOR  \op(M)  \BOR \\ 
&&  \mylet   x:T  \mybe M \myin N \BOR\\
&& \ast \BOR \tuple{M,N}_{T,U}\BOR  \ttfst_{T,U}(M)\BOR \ttsnd_{T,U}(M)  \BOR\\\
                                              && \myif B \mythen M \myelse N \BOR \\  
&&  \myletrec f(x: T):U \mybe M \myin N  \BOR \\
&& f(M) \BOR \\
                                             && \dRD{x:T.\,N}{L}{M}\\\\

 B & \BEQ  & \pred(M) \BOR \true \BOR \false
\end{array}\]
These constructs are all fairly standard, except for $\mathtt{\plainrdstring}$, which is for reverse-mode differentiation, and which we explain  below. We treat addition separately from the operations to underline the fact that the commutative monoid it forms, together with zero, is basic for differentiation. For example,  the rules for symbolic differentiation given below make essential use of this structure, but do not need any more of the available vector space structure.


    
Note the type subscripts on  pairing and projection terms. Below, we rely on these subscripts for symbolic differentiation. In practice they could, if needed, be added when type-checking.

     
     The sets $\FV(M)$ and $\FFV(M)$ of \emph{free ordinary variables} and \emph{free function variables} of a term $M$ are understood as usual (and similarly for boolean terms).  
     As is also usual, we do not distinguish $\alpha$-equivalent terms (or boolean terms). 
     
 
 
   \mylcomment{carry this tupling through to long version, includes $\tuple{\D}$}  
   
   The useful abbreviation 
     \[\mylet  \tuple{x_0\type T_0, \ldots, x_{n-1}\type T_{n-1}} \mybe M \myin N\]
     provides an elimination construct for iterated products. When $n = 0$ this is  \[\mylet x\type \unit \mybe M \myin N\] where $x \notin \FV(N)$; when $n = 1$ it is the above let construct; otherwise, it is defined recursively by:  
  \[\mylet  \tuple{x_0\type T_0, \ldots, x_{n}\type T_{n}} \mybe M \myin N = 
         \begin{array}{l}
            \mylet z \type (T_0 \times \ldots \times T_{n}) \mybe M \myin\\
            \mylet  \tuple{x_0 \type T_0, \ldots, x_{n-1}\type T_{n-1}} \mybe \ttfst_{T_0 \times \ldots \times T_{n-1} ,T_n}(z) \myin\\
            \mylet x_{n}\type T_{n} \mybe \ttsnd_{T_0 \times \ldots \times T_{n-1} ,T_n}(z) \myin N\\
          \end{array}\]    
     (where $z$ is chosen not free in $N$).

%





  We have zero and addition only at type $\real$. At other types we proceed inductively:
\[0_{\real}  =  0
\qquad 0_{\unit}  = \ast
\qquad 0_{T \times U} = \tuple{0_{T}, 0_{U}}\]
and: 
\[\begin{array}{lcl}

M +_{\unit} N & \!=\! &   

\begin{array}{l} \mylet  x \type  \unit  \mybe M \myin \\
                         \mylet  y \type \unit \mybe N \myin \ast\\[.3em]\end{array}\\\\

M +_{T \times U} N & \!=\! &   
\begin{array}{l}
 \mylet  \tuple{x_1\type  T, x_2 \type U}  \mybe M \myin\\
 \mylet  \tuple{y_1\type  T, y_2  \type U}  \mybe N \myin\\\
    \; \tuple{x_1+_T y_1,x_2+_U y_2}
\end{array}\\                                 
\end{array}\]
  
Skating over the difference between terms and their denotations,   $\dRD{x\type T.\,N}{L}{M}$ is  the reverse-mode  derivative at $L\type T$, evaluated at $M\type U$,  of the function $f$ such that $f(x:T):U = N$. 
%
Reverse-mode differentiation includes gradients as a special case.  When $T = \real^n$ and $U = \real$, the gradient of $f$ at $L$ 
is given by:
%
\[\grad{L}{ x\type \real^n.\,  N}\;  = \; \dRD{x\type \real^n.\,N}{L}{1}\]
%


For definitions of  gradients, Jacobians, and derivatives see Section~\ref{cds} below, particularly equations~(\ref{grad-eq}),~(\ref{Jac-eq}), (\ref{fdiff-eq}), and (\ref{rdiff-eq}). More generally,  for  an introduction to real analysis including vector-valued functions of several variables  and their differentials and Jacobians, see, for example,~\cite{Tr13}. 



As in~\cite{chr12}, 
and as validated by equation~(\ref{f-from-r}) below, forward-mode differentiation can be defined using nested reverse-mode differentiation. 
We can set:


%
\[\fRD{x\type T.\,N}{U,L}{M} \; = \;
\dRD{y\type U.\, \dRD{x\type T.\,N}{L}{y}}{0_U}{M} \]
 So our language effectively also has forward-mode differentiation.



Function definitions can be recursive. Indeed a function can even  be defined in terms of 
its own derivative: in a  recursive function
definition $\myletrecnospacebefore f(x\type T)\type U \mybe M \myin N$, the  language allows
occurrences of $f$  
within the term $N'$ in a sub-term $\dRD{y:T'.\,N'}{L'}{M'}$  of $M$. 
This generality may be useful
---examples 
have arisen in the context of Autograd~\cite{autograd}\footnote{See \url{ https://dougalmaclaurin.com/talk.pdf} }. Pleasantly, both our operational and denotational 
semantics accommodate it without special complications.
When we define a function without recursion, we may abbreviate $\myletrecop$ to $\myletop$.





Turning to typing, operation and predicate symbols have given  arities, 
 written  $\op\type  T \rightarrow U$ and $\pred\type  T$;
 we write $\Op_{T,U}$ for the set of operation symbols of arity $T \rightarrow U$. For example, we would have 
$ \mathit{DProd_n\type \real^{n} \times \real^{n} \rightarrow \real} $ 
and $\dlessthan\type \real^2$.
Figure~\ref{ty-rules} gives typing rules for sequents 
\[\Phi \mid \G\vdash M\type T \qquad \Phi \mid \G\vdash B\] 
where \emph{(type) environments} $\G$ have the form \[x_0\type T_0, \ldots, x_{n-1}\type T_{n-1}\] ($x_i$ all different) 
and where \emph{function (type) environments} $\Phi$ have the form 
\[f_0\type T_0 \!\rightarrow \!U_0, \ldots, f_{n-1}\type T_{n-1} \!\rightarrow\! U_{n-1}\] ($f_i$ all different).  
 We adopt the usual overwriting notations $\G[\G']$ and $\Phi[\Phi']$ for type environments.

\newcommand{\ms}[1]{\mbox{\scriptsize $#1$}}


\begin{figure}[h]

\[\hspace{-0pt}\begin{array}{c}  \ms{ \Phi \mid \G\vdash x\type T \quad  (x\type T \in \G) \qquad  \Phi \mid \G\vdash r\type\real \quad (r \in \R)}\\\\
 
  \frac{\Phi \mid \G\vdash M\,:\, \real \qquad \Phi \mid \G\vdash N\,:\, \real }{\Phi \mid \G\vdash  M + N\,:\, \real }\\\\
  
   \frac{\Phi \mid \G\vdash M\,:\,T }{\Phi \mid \G\vdash \op(M)\,:\, U} \quad  {\ms{(\op\,:\, T \rightarrow U)}}\\\\
   
   \frac{\Phi \mid \G\vdash M\,:\, T   \qquad  \Phi \mid \G[x \,:\, T] \vdash  N\,:\,U}{\Phi \mid \G \vdash \mylet  x \,:\,T \mybe M \myin N\,:\,U}\\\\
   
  { \ms{\Phi \mid \G\vdash \ast \,:\, \unit}} \qquad \frac{\Phi \mid \G\vdash M \,:\, T \qquad \Phi \mid \G\vdash N\,:\, U}{\Phi \mid \G\vdash  \tuple{M,N}_{T,U}\,:\, T \times U }\\\\
   
    \frac{\Phi \mid \G\vdash M\,:\, T \times U}{\Phi \mid \G\vdash \ttfst_{T,U}(M)\,:\, T} \qquad \frac{\Phi \mid \G\vdash M\,:\, T \times U}{\Phi \mid \G\vdash \ttsnd_{T,U}(M)\,:\, U}\\\\
    
      \frac{\Phi \mid \G\vdash B \qquad \Phi \mid \G\vdash M\,:\,T \qquad  \Phi \mid \G\vdash N\,:\,T }{\Phi \mid \G\vdash \myif B \mythen M \myelse N\,:\, T}\\\\
      
        \frac{\Phi[f\,:\, T \rightarrow U] \mid x\,:\,T \vdash M\,:\,U \qquad   \Phi[f\,:\, T \rightarrow U] \mid \G \vdash N\,:\,S }{\Phi \mid \G\vdash \myletrec f(x \,:\, T)\,:\,U \mybe M \myin N\,:\, S} \\\\
        
         \frac{\Phi \mid \G\vdash M\,:\,T}{\Phi \mid \G\vdash  f(M)\,:\, U} \quad  {\ms{(f\,:\, T\rightarrow U \in \Phi)}}\\\\
         
          \frac{ \Phi \mid \G[x\,:\,T] \vdash N\,:\,U \qquad  \Phi \mid \G\vdash L\,:\,T \qquad \;\; \Phi \mid \G\vdash M\,:\,U}{\Phi \mid \G\vdash  \dRD{x \,:\,  T.\,N}{L}{M}\,:\,T}\\\\
 
 { \ms{\Phi \mid \G \vdash \true}} \qquad \ms{{\Phi \mid \G \vdash \false}} \qquad \frac{\Phi \mid \G \vdash M\,:\,T }{\Phi \mid \G \vdash  \pred(M)} \quad {\ms{(\pred \,:\, T)}}
  \end{array}\]
  
\caption{Typing rules}
 \label{ty-rules}
   \end{figure}
The typing rule for function definitions forbids any global variable occurrences (i.e., free variables in function definitions). This restriction involves no loss in expressiveness: as in lambda lifting,  one can just add any global variables to a function's parameters, and then apply the function to the global variables wherever it is called.  The restriction enabled us to prove part (2) of Theorem~\ref{op-corr} (below), but we conjecture it is not needed.

Our various abbreviations have natural admissible typing rules: %
\[   \frac{\Phi \mid \G\vdash M \type T_0 \times \ldots, \times T_{n-1}  \quad \;\; \Phi \mid \G[x_0\type T_0, \ldots, x_{n-1}\type T_{n-1}] \vdash  N:U}{\Phi \mid \G \vdash  
\mylet  \tuple{x_0\type T_0, \ldots, x_{n-1}\type T_{n-1}} \mybe M \myin N \type U }
\]

\vspace{4pt}


%
\[\frac{\Phi \mid \G\vdash N\type \real^n \qquad  \Phi \mid \G[x\type \real^n]\vdash M\type\real}{\Phi \mid \G\vdash \grad{N}{ x\type \real^n.\,  M} \type \real^n }\]

\vspace{4pt}

\[ \frac{ \Phi \mid \G[x \type T] \vdash N \type U \quad \;\; \Phi \mid \G\vdash L \type T \quad \;\; \Phi \mid \G\vdash M \type T}{\Phi \mid \G\vdash  \fRD{x \type  T.\,N}{U,L}{M} \type U}\]

\vspace{4pt}

\[\Phi \mid \G\vdash 0_T\type T
\qquad
\frac{\Phi \mid \G\vdash M\type T \qquad \Phi \mid \G\vdash N\type T }{\Phi \mid \G\vdash M +_T N \type T }\]

We may write $\G \vdash M\type T$ (or $\vdash M\type T$) instead of $\Phi \mid \G \vdash M\type T$ if $M$  has no free ordinary (or function) variables
(and similarly for boolean terms). Typing is unique: for any $\G$, $\Phi$, and $M$ 
  there is at most one type $T$ such that $\Phi \mid \G\vdash M:T$ holds.

As an example, we use our language to program a miniature version of
an algorithm for training a machine learning model by gradient
descent, loosely based on \cite[Algorithm~8.1]{goodfellow2016book}. In
such training, one often starts with an untrained model, which is a
function from inputs (for example, images) and parameter values to
output ``predictions'' (for example, image labels). Relying on a dataset of
input/output pairs, one then picks values of the parameters by
gradient descent, as indicated in the Introduction.  
In our miniature version, we treat inputs, parameter values, and outputs as reals, 
and we assume that the training data consists of
only one fixed input/output pair $(a,b)$. 
We also assume that we have
real constants $w_0$ (for the initial value for gradient
descent), $\rate$ (for the learning rate, which is fixed) and $\maxLoss$ (for the
desired maximum loss on the dataset), 
and the infix predicate symbol $\dlessthan$.
We can then define the trained model from 
the untrained model and a loss function as follows:
\[
\begin{array}{l}
\mylet  \myTrainedModel(x \type \real) \type \real \mybe\\
\begin{array}{l}
\mylet \myLossNow(w \type \real) \type \real \mybe \myLoss(\tuple{b,\myParameterizedModel(\tuple{a,w})}) \myin\\
\mylet \gradLossNow(w \type \real) \type \real \mybe \grad{w}{w^{\prime} \type \real. \myLossNow(w^{\prime})} \myin\\
\myletrecnospacebefore \descend(w \type \real) \type \real \mybe\\  
\qquad   \myif 
\myLossNow(w) \dlessthan c 
 \mythen w 
   \myelse \descend(w - \rate * \gradLossNow(w))
   \myin\\
   \myParameterizedModel(\tuple{x,\descend(w_0)})
\end{array}\\
\myinnospacebefore \ldots
\end{array}
\]


The example above is typical of what can be expressed in our language, and many variants
of machine learning techniques that rely on gradient descent (e.g., as in  
\cite{goodfellow2016book}, and commonly used in systems like TensorFlow) are in scope as well. For
instance, there is no difficulty in expressing optimization with momentum, or differentially
private stochastic gradient descent (e.g.,~\cite{SCS13,ACG16}). 
Probabilistic choice may be treated via random number generators, as is done in practice.
Architectures that rely on convolutions or RNN cells can be expressed, even conveniently,   with a suitable choice of primitives.


 \section{Operational semantics} \label{opsem}

\newcommand{\myelsewithoutspaces}{\mathtt{else}}


 We give a big-step operational semantics,  
 specified with Felleisen and Friedman's method  using evaluation contexts and redexes \cite{FF86}. Other styles of operational semantics accommodating differentiation are surely also possible. 



Terms and boolean terms are \emph{(ordinarily) evaluated} to closed \emph{values} and (necessarily) closed \emph{boolean values}. 
The most original aspect of our operational semantics concerns the evaluation of differential terms; this is based  on the trace-based approach outlined in the Introduction, and uses a second mode of evaluation: \emph{symbolic evaluation}.

 The core idea is that to evaluate a differential term 
\[\dRD{x\type T.\,N}{L}{M}\]
 one first evaluates $L$ and $M$, and then performs  differentiation before evaluating further. There are  two  differentiation stages. First,  using the  closed value $V$ of $L$ for the \emph{differentiation variable} $x$, $N$  is \emph{symbolically evaluated} to a  trace term $C$, thereby removing all control constructs from $N$, but possibly keeping the variable $x$ free in $C$, as the derivative may well depend on it. For example, when $N$ is $\myif x \dlessthan 0 \mythen 0 \myelse x$, the value $V$ allows the guard of the conditional to be evaluated, but the occurrence
 of $x$ in the $\myelsewithoutspaces$ branch is not replaced by $V$. Second, $C$ is \emph{symbolically differentiated} at $V$ with respect to $x$.


 However, this idea is not enough by itself as the differential term may  occur inside yet another differential term. One therefore also needs to be able to symbolically evaluate the differential term. That is done much as in the case of ordinary evaluation, but now symbolically evaluating redexes in  $L$ and $M$ until one is left with the problem of symbolically evaluating a term of the form
 \[\dRD{x\type T.\,N}{V}{W}\]
 where $V$ and $W$ are values that may contain free variables. One then proceeds as above, symbolically evaluating $N$ (now using the closed value $V'$ of $V$) and then performing symbolic differentiation. As there is  some duplication between these two symbolic and ordinary evaluation processes,  our rule  for ordinarily evaluating  a differential term is designed, when executed,  to first symbolically evaluate the term, and then ordinarily evaluate the resulting trace term.
 
 The need to keep track of differentiation variables and their values for symbolic evaluation leads us to use  value environments for ordinary variables. It  is convenient to also use them for ordinary evaluation and to use function environments for function variables for both modes of evaluation.  

 Values $V,W, X, \ldots$ are terms  given by the grammar:
\[\hspace*{-1.2cm} \begin{array}{lcl}
V & \BEQ  & x \BOR  r \; (r  \in \R)  \BOR \ast \BOR  \tuple{V,W}_{T,U} 
 \end{array}\]
 %
 %
Note that, as indicated above, values may have free variables for the purposes of differentiation. Boolean values   $V_{\bool}$ are boolean terms given by: 
\[V_{\bool} \BEQ \true \BOR  \false\]
%
%
%
%
%
Closed values have unique  types $\vdash V\type T_V$;  the set of closed values of type $T$ is $\Val_T$;
and the set of boolean values is $\Val_{\bool}$.
 We assume  available operation  and predicate symbol \emph{evaluation} 
 functions 
 \[\ev: \Op_{T,U} \times \Val_T \pr \Val_U
 \qquad \quad 
 %
 %
 \bev: \Pred_{T} \times \Val_T \pr \Val_{\bool}\]
%
%
%
We also assume that for every operator $\op \type T \rightarrow U$ 
there is an operator $\rop \type T \times U \rightarrow T$.
The idea is that  $\rop(\tuple{L,M})$ is the reverse-mode  derivative of $\op$ at $L$ evaluated at $M$. We write $M.\rop(L)$ for $\rop(\tuple{L,M})$.
For example, for $\mathit{DProd_2}$ we would have:
\[\ev( \mathit{DProd}_2, \tuple{\tuple{a,b},\tuple{a',b'}}) =  aa' + bb'\]
and
\[\ev(( \mathit{DProd_2)_r}, \tuple{\tuple{\tuple{a,b},\tuple{a',b'}} ,c})=  \tuple{\tuple{ca',cb'},\tuple{ca,cb}}\]



 We  next define \emph{(value) environments} $\rho$, \emph{function environments} $\varphi$, and 
 \emph{(recursive function) closures} $\Cl$, the last two mutually recursively:

 \begin{itemize}
 \item[-] Value environments are  finite functions \[\rho = \{x_0 \mapsto V_0, \ldots, x_{n-1} \mapsto V_{n-1}\}\] from ordinary variables to closed  values.
 \item[-] Every finite function \[\varphi = \{f_0 \mapsto \Cl_0, \ldots, f_{n-1} \mapsto \Cl_{n-1}\}\] from function variables to closures is a  function environment.
  \item[-] If $\FV(M)  \subseteq   \{x\}$ and $\FFV(M)\backslash \{f\}    \subseteq     \Dom(\varphi)$
  then $\tuple{\varphi, f,x,T,U,M} $ is a closure,  written  $\fclo{\varphi}{f(x:T):  U.\, M}$.
   \end{itemize}
For any  $V$ and  $\rho$ with $\FV(V)  \hspace{-2pt} \subseteq  \hspace{-2pt}  \Dom(\rho)$, $\rho(V)$ is the closed value obtained by substituting $\rho(x)$ for all free occurrences of $x$ in $V$. 
 
\emph{Trace terms} $C,D,\ldots$,  are defined as follows:
 \[\begin{array}{lcl}
C \hspace{-2pt} & \BEQ \hspace{-2pt} & x  \BOR   r \;\; (r\in \R) \BOR C + D  \BOR  \op(C) \BOR \\
                                               &&  \mylet   x \type T  \mybe C \myin D \BOR\\
                                             && \ast \BOR \tuple{C,D}_{T,U}\BOR   \ttfst_{T,U}(C)\BOR \ttsnd_{T,U}(C) 
                                             \end{array}\]
 They are the  terms with no conditionals, function definitions or applications, or   differentiations.
 

We will define two \emph{ordinary} evaluation relations, and one \emph{symbolic} one:

 
 \begin{itemize}

 \item  For all $\varphi$ and $\rho$ we define  evaluation relations  between terms and closed values 
and  between boolean terms and closed boolean values 
 via rules establishing sequents of the forms:
\vspace{-4pt}
\[\vspace{-4pt}\varphi \mid \rho  \vdash M \impe V \qquad\quad  \varphi \mid \rho  \vdash B \impe V_{\bool}\]
%

%
%
%
 \item  For all $\varphi$ and $\rho$ we define a symbolic evaluation relation  between terms and trace terms via rules establishing sequents of the form:
\vspace{-4pt}
\[\vspace{-4pt}\varphi \mid \rho  \vdash M \imps C\]
%

 \end{itemize}

%

 %
%
 \emph{Evaluation contexts (boolean evaluation contexts)}, ranged over by $E$ (resp.\ $E_{\bool}$ ), 
are terms with a unique \emph{hole} $\hole$ : 
 \[\begin{array}{ccl}
E \hspace{-5pt} &  \BEQ  \hspace{-5pt} & \hole \BOR E + N  \BOR  V + E \BOR \op(E) \BOR \\

                                               &&  \mylet   x \type T  \mybe E \myin N \BOR\\
                                             &&  \tuple{E,N}_{T,U} \BOR \tuple{V,E}_{T,U} \BOR 
                                              \ttfst_{T,U}(E)\BOR \ttsnd_{T,U}(E)  \BOR\\ 
                                            && \myif E_{\bool} \mythen M \myelse N \BOR \\
                                           &&  
                                           f(E) \BOR \\
                                             && \dRD{x:T.\,N}{E}{M} \BOR \dRD{x:T.\,N}{V}{E}\\\\
                                             
 E_{\bool}\hspace{-5pt}  & \BEQ \hspace{-5pt} &  \pred(E)
 \end{array}\]
We write $E[M]$ for the term obtained by replacing the hole $\hole$ in $E$ by the term $M$ and $E_{bool}[M]$ similarly;
a context $E$ is \emph{trivial} if it is $[\;]$; and
 $\FV$ and $\FFV$ are extended to contexts. 
We have $\FV(E[M]) = \FV(E) \cup \FV(M)$ and $\FFV(E[M]) = \FFV(E) \cup \FFV(M)$ and similarly for boolean contexts.
 
 
 
 \emph{Redexes}, ranged over by $R$, and  \emph{boolean redexes}, ranged over by $R_{bool}$, 
 are given by: 
 
  \[\hspace{-5pt}\begin{array}{lcl}
R  \hspace{-10pt}  &\hspace{-5pt} \BEQ  &   V + W \BOR \op(V)\BOR\\

                                               &&  \mylet   x\type T  \mybe V \myin N \BOR\\
                                             &&   \ttfst_{T,U}(V)\BOR \ttsnd_{T,U}(V)  \BOR\\
                                           && \myif V_{\bool} \mythen M \myelse N \BOR \\
&&\!\!   \myletrec f(x\type  T)\type U \mybe M \myin N \BOR   f(V) \BOR \\
                                             && \dRD{x\type T.\,N}{V}{W} \\\\
R_{\bool} &\hspace{-5pt} \BEQ  & \pred(V)
 \end{array}\]
 Note that boolean expressions are useful here in that they enable separate conditional and predicate redexes, and so  evaluating predicates and making choices are distinct in the operational semantics.
 

The next lemma is the basis of a division into cases that supports  operational semantics using evaluation contexts 
in the style of Felleisen and Friedman.
 \begin{lemma}[Evaluation context analysis] \label{context-form-anal}
   \hspace{2em}
  \begin{enumerate}
  \item  
  Every  term $M$, other than a value, has exactly one of the following two forms:
 \begin{itemize}
  \item  $E[R]$ for a unique  evaluation context 
  and redex, or  
 \item  $E[R_{\bool}]$  for a unique, and non-trivial,  evaluation context  and boolean redex. 
\end{itemize}
\item
  Every  boolean term $B$, other than a boolean value, has exactly one of the following two forms:
 \begin{itemize}
  \item  $E_{bool}[R]$ for a unique, and non-trivial,  boolean evaluation context 
  and redex, or  
 \item  $E_{bool}[R_{\bool}]$  for a unique boolean evaluation context  and boolean redex. 
\end{itemize}
 \end{enumerate}
\end{lemma}
  \POPLomit{ \begin{proof} The proof 
 is a straightforward simultaneous structural induction.
 \end{proof}}

The next lemma is useful to  track types when proving theorems about the operational semantics.
  \begin{lemma}[Evaluation context polymorphism]\label{ctx-type}
    Suppose that $\Phi \mid \G \vdash E[M] \type T$. Then, for some type $U$ we have $\Phi \mid \G \vdash M \type U$ and, whenever $\Phi \mid \G \vdash N \type U$, we have $\Phi \mid \G \vdash E[N] \type T$. 
    
    Analogous results hold for typings of any of the forms $\Phi \mid \G \vdash E[B]:T$ or  $\Phi \mid \G \vdash E_{\bool}[M]$ or $\Phi \mid \G \vdash E_{\bool}[B]$.
  \POPLomit{
  \hspace{2em}
  \begin{enumerate}
  \item  
  Suppose that $\Phi \mid \G \vdash E[M] \vdash T$. Then, for some type $U$ we have $\Phi \mid \G \vdash M \type U$ and, whenever $\Phi \mid \G \vdash N \type U$, we have $\Phi \mid \G \vdash E[N] \type T$.
  \item  Suppose that $\Phi \mid \G \vdash E[B]:T$. Then we have $\Phi \mid \G \vdash B$, and, whenever $\Phi \mid \G \vdash B'$, we have $\Phi 
  \mid \G \vdash E[B']:T$.
    
  \item Suppose that $\Phi \mid \G \vdash E_{\bool}[M]$. Then, for some type $U$ we have $\Phi \mid \G \vdash M \type U$ and, whenever $\Phi \mid \G \vdash N \type U$, we have $\Phi \mid \G \vdash E_{\bool}[N]$.
  
  \item  Suppose that $\Phi \mid \G \vdash E_{\bool}[B]$. Then we have  $\Phi \mid \G \vdash B$, and, whenever $\Phi \mid \G \vdash B'$, we have $
  \Phi \mid \G \vdash E_{\bool}[B']$.
  \end{enumerate}
  }
  \end{lemma}
    \POPLomit{ \begin{proof}
  Parts (1) and (3) are proved by simultaneous structural induction, as are parts (2) and (4). Part (2) of Proposition~\ref{UT} is used in the proof 
  of parts (1) and~(3).
   \end{proof}}
 
   By the uniqueness of types, the types whose existence is claimed in the above lemma are unique.

\begin{figure}[h] 
\vspace{-15pt}
\begin{multicols}{2}

\[\begin{array}{l}

\ms{ \varphi \mid \rho  \vdash  V \impe \rho(V)}\\\\

   \frac{ \varphi \mid \rho \vdash V \impe r \quad \varphi \mid \rho  \vdash W \impe s}{\varphi \mid \rho  \vdash V + W \impe t} \\ \qquad \ms{(\mbox{where } t = r + s)}\\\\

       \frac{\varphi \mid \rho  \vdash V \impe V'}{\varphi \mid \rho  \vdash  \op(V) \impe W } 
       \\ \qquad \ms{(\mbox{where } \ev(\op,V')  \simeq  W)}\\\\

   \frac{\varphi \mid \rho  \vdash V \impe V' \quad \varphi \mid \rho[V'/x]  \vdash N \impe W}
   {\varphi \mid \rho  \vdash   \mylet   x \type T  \mybe V \myin N \impe W}
   \\\\

     \frac{\varphi \mid \rho  \vdash V \impe \tuple{W_1,W_2}_{T,U}}
             {\varphi \mid \rho  \vdash \ttfst_{T,U}(V) \impe W_1}\\\\
             
             \frac{\varphi \mid \rho  \vdash V \impe \tuple{W_1,W_2}_{T,U}}
             {\varphi \mid \rho  \vdash \ttsnd_{T,U}(V) \impe W_2}\\\\

    \end{array}\]

\columnbreak 
 \[\begin{array}{l}
 \frac{\varphi \mid \rho  \vdash M \impe V}{\varphi \mid \rho  \vdash \myif \true \mythen M \myelse N \impe V}\\\\
 
 \frac{\varphi \mid \rho  \vdash N \impe V}{\varphi \mid \rho  \vdash \myif \false \mythen M \myelse N \impe V}\\\\

   \frac{\varphi[\fclo{\varphi}{f(x:T):  U.\, M}/f] \mid \rho \vdash N \impe V }
                     {\varphi \mid \rho \vdash \myletrec f(x: T):U \mybe M \myin N \impe  V}
            \\\\

           \frac{\varphi \mid \rho \vdash V \impe V' \quad
            \varphi'[\varphi(f)/f]\mid \{x \mapsto V'\} \vdash M \impe W}
                     {\varphi \mid \rho \vdash f(V) \impe  W}\\
          \quad  \ms{(\mbox{where  $\varphi(f) = \fclo{\varphi'}{f(x:T):  U.\, M}$})}\\\\

 \frac{
\varphi \mid \rho  \vdash \dRD{x:T.\,N}{ V}{W} \imps C \quad \varphi \mid \rho  \vdash C \impe X}{\varphi \mid \rho  \vdash \dRD{x:T.\,N}{ V }{W} \impe X}\\\\

       \frac{\varphi \mid \rho  \vdash V \impe V'}{\varphi \mid \rho  \vdash  \pred(V) \impe W_{\bool} } 
       \\ \qquad {\ms{(\mbox{where } \bev(\pred,V')  \simeq  W_{\bool})}}
       \end{array}\]
       
    \end{multicols}
       
 \caption{Ordinary operational semantics: values and redexes}
 \label{op-sem-vr}
   \end{figure}    
   
   \begin{figure}[h]

\vspace{11.5pt}
   \[\hspace{-3pt}\begin{array}{l}

               \frac{\varphi \mid \rho \vdash R  \impe V \quad
                       \varphi \mid \rho[V/x] \vdash E[x] \impe W}
                      {\varphi \mid \rho \vdash E[R] \impe W}                                       
                                      \,\quad \ms{(\mbox{$E$ nontrivial, $x \notin  \Dom(\rho)$})}
              \\\\

         \frac{\varphi \mid \rho \vdash R_{\bool}  \impe V_{\bool} \quad
                       \varphi \mid \rho \vdash E[V_{\bool}] \impe W}
                      {\varphi \mid \rho \vdash E[R_{\bool}] \impe W} 
              \\\\

                \frac{\varphi \mid \rho \vdash R  \impe V \quad
                       \varphi \mid \rho[V/x] \vdash E_{\bool}[x] \impe W_{\bool}}
                      {\varphi \mid \rho \vdash E_{\bool}[R] \impe W_{\bool}} 
                                                            \qquad \ms{(x \notin  \Dom(\rho))}
              \\\\

         \frac{\varphi \mid \rho \vdash R_{\bool}  \impe V_{\bool} \quad
                       \varphi \mid \rho \vdash E_{\bool}[V_{\bool}] \impe W_{\bool}}
                      {\varphi \mid \rho \vdash E_{\bool}[R_{\bool}] \impe W_{\bool}}

   \end{array}\]     

                \caption{Ordinary operational semantics: contexts}
                              \label{op-sem-c}
   \end{figure}      

\begin{figure}[h]

\vspace{-15pt}
\begin{multicols}{2}
\[\begin{array}{l}

  \ms {\varphi \mid \rho  \vdash  V \imps V}\\\\

 \ms{ \varphi \mid \rho  \vdash V + W  \imps V + W}\\\\

     \ms{\varphi \mid \rho  \vdash \op(V) \imps \op(V)}\\\\

        \frac{\varphi \mid \rho  \vdash  V \impe V' \quad \varphi \mid \rho[V'/x]  \vdash N \imps C}
             { \varphi \mid \rho  \vdash \mylet   x :T  \mybe V \myin N  \imps   \mylet   x : T  \mybe V \myin C }
    \\\\

\ms{  \varphi \mid \rho  \vdash \ttfst_{T,U}(V) \imps \ttfst_{T,U}(V)}\\\\

\ms{  \varphi \mid \rho  \vdash \ttsnd_{T,U}(V) \imps \ttsnd_{T,U}(V)}\\\\ 
 
       \end{array}\]
 
 \columnbreak
 
 \[\begin{array}{l}

 \frac{\varphi \mid \rho  \vdash M \imps C}{\varphi \mid \rho  \vdash \myif \true \mythen M \myelse N \imps C}\\\\
 
 \frac{\varphi \mid \rho  \vdash N \imps C}{\varphi \mid \rho  \vdash \myif \false \mythen M \myelse N \imps C}\\\\

          \frac{\varphi[\fclo{\varphi}{f(x:T):  U.\, M}/f] \mid \rho \vdash N \imps C }
                     {\varphi \mid \rho \vdash \myletrec f(x: T):U \mybe M \myin N \imps  C}
            \\\\

             \frac{\varphi \mid \rho \vdash V \impe V' \quad
                       \varphi'[\varphi(f)/f]\mid \{x  \mapsto V'\} \vdash M \imps C}
                     {\varphi \mid \rho \vdash f(V) \imps   \mylet x \type T \mybe V \myin C}
            \\ \quad \ms{ (\mbox{where  } \varphi(f) = \fclo{\varphi'}{f(x:T):  U.\, M}})\\\\

                  \frac{\varphi \mid \rho \vdash  V \impe V' \quad
                          \varphi \mid \rho[V'/x] \vdash  N \imps C}
                       {\varphi \mid \rho \vdash  \dRD{x:T.\,N}{V}{W} \imps  \newcalRprop{x :  T }{C}{V}{W}{}} 
              \\\\
                                 \end{array}\]  
                                 
                                   \end{multicols}  
                                   \vspace{-12.5pt}   
                \caption{Symbolic operational semantics: values and redexes} 
                \label{op-sem-s-vr}
                    \end{figure}

\begin{figure}[h]

\[\begin{array}{l}
             
             \frac{\varphi \mid \rho \vdash R \imps C \quad
                       \varphi \mid \rho \vdash C \impe V \quad
                       \varphi \mid \rho[V/x] \vdash E[x] \imps D}
                      {\varphi \mid \rho \vdash E[R] \imps \mylet x:T_V \mybe C \myin D } 
                            \quad \ms{(\mbox{$E$ nontrivial and $x \notin  \Dom(\rho)$})}\\\\

             \frac{\varphi \mid \rho \vdash R_{\bool} \impe V_{\bool} \quad
                       \varphi \mid \rho \vdash E[V_{\bool}] \imps C}
                      {\varphi \mid \rho \vdash E[R_{\bool}] \imps C }

                      \end{array}\]
                                      
       \caption{Symbolic operational semantics: contexts}
\label{op-sem-s-c}
            \end{figure}      

\newcommand{\newcalRfig}[4]{#4\hspace{-0.75pt} .\hspace{-0.5pt}\mathcal{R}_{#3} (#1\hspace{-1pt}.\hspace{0.75pt}#2 )}

\begin{figure}[h]

\[{\hspace{-6pt}\begin{array}{lll}

   \newcalRfig{     x\type T   }{y}{V}{W} & = &  \left \{\begin{array}{ll} W & (y = x)\\ 0_T & (y  \neq  x)
                                                           \end{array} \right .\\[10pt]
                                                           
                          
          \newcalRfig{    x\type T  }{r}{V}{W} & = &  0_T \qquad (r \in \R)\\[10pt]
          
                                      
                          \newcalRfig{    x\type T  }{\hspace{-0.5pt}D \!+ \! E}{V}{W} & = &   
                                            \newcalRfig{    x\type T  }{D}{V}{W} \!+_T\! \newcalRfig{    x\type T  }{E}{V}{W}\\[10pt] 

                                               
                                                                             
                           \newcalRfig{    x\type T  }{\op(D)}{V}{W} & = &   
                                                                             \begin{array}{l}\mylet     x\type T \,   \mybe V \myin\\
                                                                             \mylet y \type S\, \mybe W.\rop(D) \myin   \newcalRfig{   x\type T }{D}{V}{y} 
                                                                   \end{array}   \\ && \qquad (y \notin \FV(V), \op\type  S \rightarrow U )
                                                                            \\[10pt]

%


                                                                             %

    \newcalRfig{   x\type T }
                      {\mylet   y \type S \! \mybe\! D \myin E}{V}{W} & = &                                                 
                                                                             
        
 \hspace{-4pt} \begin{array}{l} 
   \mylet     x  \type T \!  \!\mybe\! V \myin\\
   \mylet     y \type  S\! \mybe \! D \myin\\
\;\;\newcalRfig{   x\type T  }{E} { V}{W}  \; +_T   \; \\
\;\;(\mylet  \ov{y}\type  S\! \mybe\! \newcalRfig{ y\type S }{E} { y}{W} \! \myin\! \, \newcalRfig{   x\type T }{D}{V}{\ov{y}})
  \end{array}   \\
&&   \hspace{20pt} (y \notin \FV(W), y,\ov{y}\notin \FV(V,D))
  
  \\[10pt]   
  
                                                                             
        
  
  
    %
  


          \newcalRfig{   x\type T  }{ \ast}{V}{W} & = &  0_T\\[10pt]     


           \newcalR{    x\type T  }{\tuple{D,E}_{U,S} }{V}{W} & = &   
                \begin{array}{l} \mylet y \type U, z \type S \mybe W \myin\\ \; \newcalR{    x\type T  }{D}{V}{y} +_T \newcalR{ x\type T }{E}{V}{z} \\
                \quad\quad  (y,z \notin \FV(V,D,E) 
                )
                \end{array}\\[10pt] 


  
             \newcalRfig{    x\type T  }{ \ttfst_{U,S}(D)\hspace{-0.5pt}}{V}{W} & = &  
                                                                                             \begin{array}{l}
                                                                                               \mylet     x\type T \,   \mybe V \myin\\
                                                                                                \mylet y\type  U \times S \! \mybe \! D \myin  \newcalRfig{    x\type T  }{ D}{V}{\tuple{W,0_{S}}}
                                                                                               \end{array} 
                                                                       \\&&\hspace{30pt} (y \notin \FV(V,W,D))\\[10pt]


  
   \newcalRfig{    x\type T  }{ \ttsnd_{U,S}(D)\hspace{-0.5pt}}{V}{W} & = &  
                                                                                             \begin{array}{l}
                                                                                               \mylet     x\type T \,   \mybe V \myin\\
                                                                                                \mylet y\type  U \times S \! \mybe \! D \myin \newcalRfig{    x\type T  }{ D}{V}{\tuple{0_{U},W}}
                                                                                               \end{array} 
                                                          \\&&\hspace{30pt} (y \notin \FV(V,W,D))\\[10pt]


 \end{array}}\]
 
 \mylcomment{bug corrected in let case (adding $y \notin \FV(V)$). \\
 Bug corrected in tuple case; old version commented out above\\
Variables renamed in tuple case\\
 Check in proofs!}
 \caption{Definition  of $  \newcalR{  x\type T }{C}{V}{W}$}
 \label{for-diff}
 \end{figure}
    %


The rules for ordinary evaluation are given in Figures~\ref{op-sem-vr} and~\ref{op-sem-c}; those for symbolic evaluation are given in Figures~\ref{op-sem-s-vr} and~\ref{op-sem-s-c}. The definitions are mutually recursive. They make use of the symbolic differentiation of trace terms: given a trace term $C$, and  values $V$ and $W$ (not necessarily closed), 
     we  define a trace term
            \[\newcalRprop{ x \type T}{C}{V}{W}{}\]
 intended to denote the reverse-mode  derivative of the function  $x \type T \mapsto C$, at $V$, evaluated at $W$. A definition is given  in Figure~\ref{for-diff}; 
 %
            %
in the definition we assume that $x \notin \FV(V,W)$, 
and, as is common, that all binding variables 
are different.

\mylcomment{make sure these remarks carry over to the long version, including the SCT}

 \begin{proposition} \label{symb-diff-type}
The following  typing rule is admissible:
   \[\frac{ \G[x\type T] \vdash C \type  U \quad   \G \vdash V \type T \quad  \G  \vdash W \type U}{ \G  \vdash \newcalRprop{ x\type T}{C}{V}{W}{} \type  T}\]
\end{proposition}  
\POPLomit{\begin{proof}
By induction on $C$, using the definition of $\newcalRprop{x\type T}{C}{V}{W}{}$ for the various cases.
\end{proof}}


In large part because of the restrictions on trace terms,
their symbolic differentiation is just a systematic, formal application of the chain rule.
In our setting, this application requires a fair amount of attention to detail, for instance 
the use of the type decorations when giving derivatives of pairing and projection terms. 

The reader may wish to try the following two evaluation examples with nested differentiation:
 \[\dRD{x \type \real.\, x\, \times \, \dRD{y \type \real.\, x + y}{1}{1}}{1}{1} \impe 1 \]
 and 
 \[ \myletrec \; f(x \type  \real ) \type  \real  \mybe \dRD{y \type  \real.\,x+y}{1}{1} \; 
 \myin \; \dRD{x \type  \real .x + f(x)}{1}{1} \impe 1\]
 Examples of this kind can be used to illustrate perturbation confusion in forward differentiation, e.g.,~\cite{SP05,SP08}.


We need some basic results on our evaluation relations. Two are standard: determinacy and type safety, and are used implicitly throughout the rest of the paper. The third connects symbolic and ordinary evaluation: one can interpolate  symbolic evaluation within ordinary evaluation. It is principally helpful to reduce the completeness part of symbolic evaluation to the completeness of ordinary evaluation (see Theorem~\ref{op-com}).

\begin{proposition}[Determinacy of evaluation] \label{ev-det} The following hold: 
\hspace{2em}
\begin{enumerate}
\item For any $\varphi$, $\rho$, and $M$, there is at most one  value $V$ s.t.\  $\varphi \mid \rho \vdash M \impe  V$.
\item For any $\varphi$, $\rho$, and $M$, there is at most one trace term $C$ s.t.\  $\varphi \mid \rho \vdash M \imps  C$.
\end{enumerate}
\end{proposition}
\POPLomit{ \begin{proof} The two statements are proved by simultaneous structural induction on the sizes of the proofs of each of judgements, and 
by cases on the forms of the terms, making use of Lemma~\ref{context-form-anal} to distinguish the various cases, and 
Lemma~\ref{a-equivalence} to handle cases involving ordinary variable binding.
\end{proof}}





The following interpolation proposition  establishes a certain consistency between the ordinary and symbolic evaluation relations.
\begin{proposition}[Operational interpolation] \label{ord-symb-con}
For all $\varphi$, $\rho$, and closed values $V$, the following are equivalent: 
\begin{itemize}
\item[(1)] $\varphi \mid \rho \vdash M \impe V$,
\item[(2)] $\varphi \mid \rho \vdash M \imps  C$ and $\varphi \mid \rho \vdash C \impe  V$, for some $C$.
\end{itemize}


\end{proposition}  
\POPLomit{ \begin{proof} 
The proofs are by  induction on the size of the proofs of $\varphi \mid \rho \vdash M \impe V$ or $\varphi \mid \rho \vdash M \imps  C$
The implication from (1) to (2) is proved by  induction on the size of the proof that $\varphi \mid \rho \vdash M \impe V$; 
the converse is proved by  induction on the size of the proof that $\varphi \mid \rho \vdash M \imps  C$.
In each case, both function redex cases make use of Lemma~\ref{general}.

\end{proof}}

For a type safety theorem, we  need  typing judgments $\rho \type  \G$, $\varphi \type  \Phi$, and 
$\Cl\type  T\rightarrow U$ for environments, function environments, and closures (implicitly extending the notion of type). These are defined inductively by the following rules:
%

\[\frac{\vdash V_i \type  T_i \quad (i = 0,n-1)}
{\vdash \{x_0 \mapsto V_0,\ldots, x_{n-1}\mapsto V_{n-1}\} \;\type \; x_0: T_0,\ldots, x_{n-1}: T_{n-1}}\]

 \vspace{4pt}
 
 
 \[\frac{\vdash \Cl_i \type  T_i \rightarrow U_i \quad (i = 0,n-1)}
 {\vdash\{f_0 \mapsto \Cl_0, \ldots, f_{n-1} \mapsto  \Cl_{n-1}\}\; \type \; f_0:T_0 \!\rightarrow \!U_0, \ldots, f_{n-1}:T_{n-1} \!\rightarrow\! U_{n-1}}\]
 
\vspace{4pt}

 \[\frac
 {\vdash \varphi: \Phi  \quad \Phi[f\type T \rightarrow U] \mid  x:T\vdash M\type  U}
 {\vdash \fclo{\varphi}{f(x:T):  U.\, M}\,\type \, T \rightarrow U}
  \]
  %
%
Note that a closure $\fclo{\varphi}{f(x:T):  U.\, M}$ can only have type $T \rightarrow U$. So in the third rule $\Phi$ is determined up to the ordering of its function type declarations. Whether the conclusion of the rule follows does not depend on the choice of this ordering.  We write $\Phi_{\varphi}$ for the choice of $\Phi$ with  declarations ordered using the standard function variable ordering.

\begin{proposition}[Type safety] \label{type-safety} Suppose   $\Phi\!\mid\! \G \!\vdash\! M\type T$, $\vdash \varphi\type  \Phi$ and $\vdash \rho\type \G$. 
Then we have:
\[ \varphi\mid \rho\vdash M \;\impe\; V \implies \vdash V\type T \]
and
 \[ \varphi\mid \rho \vdash M \;\imps\; C \implies \G  \vdash C\type T \]
\end{proposition}
\POPLomit{ \begin{proof} 
A straightforward induction on values $W$ shows that if $\vdash \rho\type \G$, $ \G \vdash W\type T$, and $\rho\vdash W \;\impe\; V$ then $ \G 
\vdash V:T$. The proof is then a simultaneous induction on the sizes of the proofs of the two judgements. In the case of evaluation 
contexts, the proof makes use of Lemma~\ref{ctx-type}. The proof of the second implication uses Proposition~\ref{symb-diff-type} in the 
case of reverse-mode differentiation redexes.
\end{proof}}

\section{Mathematical preliminaries}\label{math}


\newcommand{\cX}{\mathcal{X}}
\newcommand{\cY}{\mathcal{Y}}


We now turn to the mathematical facts needed for the denotational semantics of our language. These concern the two modes of differentiation and their interaction with domain theory. We follow~\cite{AJ94} for domain theory, but write dcppo for pointed dcpo, and say
a partial order is \emph{coherent} iff every compatible subset has a lub.   (A subset is \emph{compatible} if any two of its elements have an upper bound.) 
Every coherent partial order is a dcppo.


 
 The collection of  partial functions $f \type   X \pr Y$ with open domain between two topological spaces forms a partial order under  graph inclusion:
\[f \leq g \iff f \subseteq g\]
equivalently,  using the Kleene order\!
\footnote{We write $e \preceq e'$ for two mathematical expressions $e$ and $e'$ to mean that if  $e$ is defined so is $e'$, and they are then equal.}:
\[f \leq g \iff \forall \bx \in X.\, f\bx \preceq g\bx\]
This partial order is a coherent dcppo with $\perp$ the everywhere undefined function and compatible sups given by unions. A partial function $f \type   X \pr Y$ is \emph{continuous} if $f^{-1}(B)$ is open whenever $B$ is; the subcollection of continuous partial functions forms a coherent subdcppo.
%
This holds as, for any open set $B \subseteq Y$
and compatible collection of partial functions $f_i \; (i \in I) \type   X \pr Y$,  we have:
\[(\bigvee_{i \in I} f_i)^{-1}(B) = \bigcup_{i \in I} f_i^{-1}(B)\]
%
We write $\cont{X}{Y}$ for the dcppo of partial continuous functions from $X$ to $Y$.

 It is convenient to use a variation on cartesian product when working with powers of $\R$. We set:
\[\R^m \dtimes \R^n \;\eqdef\; \R^{m+n} \quad (m,n \geq 0)\]
This version of product is associative. Vector concatenation then serves as tupling; however, for clarity,  we may use the usual notation $(\bx_0,\ldots, \bx_{k-1})$ instead of 
$\bx_0\ldots \bx_{k-1}$. There are evident definitions of  the projections 
$\pi_i^{m_0,\ldots, m_{k-1}}\type  \R^{m_0} \dtimes \ldots \dtimes\R^{m_{k-1}} \rightarrow \R^{m_i}$, and of the tupling  
\[\tuple{f_0,\ldots, f_{k-1}}\type  \R^n \pr  \R^{m_0} \dtimes \ldots \dtimes\R^{m_{k-1}} \]
 of  $f_i\type  \R^n \pr \R^{m_i}$. We may ignore the superscripts on the projections when they can be understood from the context.


\subsection{Continuity, Differentiability, and Smoothness} \label{cds}

Standard multivariate analysis of vector-valued real functions from $\R^n$ to $\R^m$ (with $n,m >0$) considers functions $f$ defined on an open domain of $\R^n$, see, e.g.,~\cite{Tr13}. These are precisely  the partial functions:
\[f \type   \R^n \pr \R^m \quad (n,m >0)\]
with open domain. 

%

When $m \!= \!1$, such a function $f$ has partial derivatives 
\[\partial_j(f) \type   \R^n \pr \R \qquad (j = 0, n-1)\]
where
\[\partial_j(f)(x_0,\ldots,x_{n-1})  \simeq  _{\mathrm{def}} \frac{\partial f}{\partial x_j}\]
viewing $f$ as a function of $x_0,\ldots,x_{n-1}$. 
Taken together, these partial derivatives form its \emph{gradient} 
\[\nabla(f) \type   \R^n \pr \R^n \]
where
%
\begin{equation}\label{grad-eq}
\nabla(f)(\bx)  \simeq   \tuple{\partial_0(f)(\bx), \ldots, \partial_{n-1}(f)(\bx)}\
\end{equation}
%
We write $\nabla_{\bx}(f)$ for $\nabla(f)(\bx)$.

We say  $f$  is \emph{continuously differentiable} if  all the 
$\partial_j(f)$ are continuous with domain that of $f$,  equivalently if  $\nabla(f)$ is continuous with domain that of $f$.  As 
an example, removing $0$ from the domain of definition of the non-differentiable ReLU function $f(x) = \max(x,0)$, 
we obtain a continuously differentiable partial function with domain $\R \backslash  0$; its derivative also has 
domain $\R \backslash  0$, with  value $0$, if $x < 0$, and $1$, if $x > 0$.




 
 We now turn to the general case where $f \type   \R^n \pr \R^m$.  The  \emph{Jacobian} 
 $\jacobian_{\bx}(f)$ of $f$ at $\bx \in \R^n$  is the $m$ by $n$ matrix: 
 \begin{equation}\label{Jac-eq}
 \jacobian_{\bx}(f)_{i,j} \simeq \partial_{j}(\pi_i\circ f)(\bx) 
 \end{equation}
where the matrix is undefined if any of the $\partial_{j}(\pi_i\circ f)(\bx) $ are. 
%
 %
 These Jacobians form a partial function:
 \[ \jacobian(f) \type   \R^n \pr \Mat(m,n)\]
where $\Mat(m,n)$ is  the collection  of $m$ by $n$ matrices. Viewing $\Mat(m,n)$ as  $\R^{m \times n}$,
we say 
$f$ is  \emph{continuously differentiable}  if $\jacobian(f)$ is continuous and  has the same domain as $f$ (equivalently 
if each component $\pi_i \hspace{-0.65pt}\circ\hspace{-0.65pt} f$ of $f$ is continuously differentiable).

 The \emph{ differential}
 \footnote{Differentials are discussed in~\cite{Tr13} see: p325 for differentials of functions of several variables;
p348 for higher-order differentials; 
p381 for  differentials of vector-valued functions of several variables; and
p388 for  the chain rule in differential terms.}
 \[\myd(f) \type   \R^n \dtimes \R^n \pr \R^m\]
   of $f$  is defined by:
  \begin{equation}\label{fdiff-eq}
  \myd(f)(\bx, \by)  \simeq  \jacobian_{\bx}(f)\cdot\by
  \end{equation}
  %
 %
and $f$ is continuously differentiable iff $\myd(f)$ has domain $ \Dom(f) \times \R^n$ and is continuous there. 


 We write $\myd_{\bx}(f)$ for the partial function $\myd(f)(\bx,-)$; it is either  $\perp$ or everywhere defined and linear, the latter occurring precisely when $\jacobian_{\bx}(f)$ is defined.
If $f$ is linear then $\myd_{\bx}f = f$, for all $\bx \in \R^n$.



 In the automatic differentiation literature,  $\myd_{\bx}(f)$ is called the \emph{forward-mode} derivative of $f$ at $\bx$. For the \emph{reverse-mode} derivative we 
 define:
  \[\mydr(f) \type   \R^n \dtimes \R^m \pr \R^n\]
 by:
   \begin{equation}\label{rdiff-eq}
   \mydr(f)(\bx, \by)  \simeq  \jacobian_{\bx}(f)^t\cdot\by
   \end{equation}
 and write $\mydr_{\bx}(f)$ for the partial function $\mydr(f)(\bx,-)$; 
 %
 $f$ is continuously differentiable iff $\mydr(f)$ has domain $ \Dom(f) \times \R^m$ and is continuous there.


 In terms of the differentials, the two modes are related by:
  \begin{equation}\label{mode-rel}
  \mydr_{\bx}(f) \, = \, \myd_{\bx}(f)^{\dagger}
  \end{equation}
(setting $\perp^{\dagger} = \perp$). So, if $f$ is linear, then $\mydr_{\bx}(f) = \myd_{\bx}(f)^{\dagger} = f^{\dagger}$.

The semantic content of the definition of forward-mode from reverse-mode given in Section~\ref{language} is the  following equality:
  \begin{equation}\label{f-from-r}
  \myd_{\bx} f = \myd^R_{\mathbf{0}}(\myd^R_{\bx}f)
  \end{equation}
This holds as: $\myd^R_{\mathbf{0}}(\myd^R_{\bx}f) = (\myd^R_{\bx}f)^{\dagger} =  ((\myd_{\bx}f)^{\dagger})^{\dagger} = \myd_{\bx}f$  (the first equality holds as $\myd^R_{\bx}f$ is linear if $\neq \perp$).


The continuously differentiable functions are closed under composition. If  $h\type \R^n \pr \R^l$ is the composition of two such  functions $f\type \R^n \pr \R^m$  and $g\type\R^m \pr \R^l$, then  the chain rule expresses the derivative of $h$ in terms of those of $f$ and $g$. In terms of Jacobians, the chain rule is:
\[\begin{array}{lcll} \jacobian_{\bx}(h)&  \simeq  & \jacobian_{f(\bx)}(g)\cdot \jacobian_{\bx}(f) & (\bx \in   \Dom(h))\\
\end{array}\]
(Note that $\bx \in   \Dom(h)$ iff $\bx \in   \Dom(f)$ and $f(\bx) \in   \Dom(g)$.)
In terms of forward-mode derivatives the chain rule is:
  \begin{equation}\label{for-chain}
  \begin{array}{lcll} \myd_{\bx}(h) &  =  & \myd_{f(\bx)}(g)\circ \myd_{\bx}(f) & (\bx \in   \Dom(h))\\
\end{array}
\end{equation}
and in terms of reverse-mode derivatives it is:
  \begin{equation}\label{rev-chain}
  \begin{array}{lcll} \mydr_{\bx}(h) &  =  & \mydr_{\bx}(f)\circ\mydr_{f(\bx)}(g) & (\bx \in   \Dom(h))\\
\end{array}
\end{equation}
%
%
%




Derivatives with respect to two variables can be reduced to derivatives in each separately. Specifically, suppose 
$f \type \R^{n + n'} \pr \R^m$. Then, for $\bx \in \R^n$ and $\by \in \R^{n'}$, we have:
\begin{equation}\label{double-forward}
 \myd_{\tuple{\bx,\by}}(f)(\bu,\bv) \;\simeq\; \myd_{\bx}(f(-,\by))(\bu) + \myd_{\by}(f(\bx,-))(\bv) \quad (\bu \in \R^n, \bv \in \R^{n'})
\end{equation}
and
\begin{equation} \label{double-reverse} \myd^R_{\tuple{\bx,\by}}(f) = \tuple{\myd^R_{\bx}(f(-,\by)), \myd^R_{\by}(f(\bx,-))}
\end{equation}
These equations are useful for dealing with fan-in. 


Our programming language  has all finite product types of the reals. We will therefore need to work with partial  functions  with open domain 
$f \type   \R^n \pr \R^m$ where $n$ or $m$  is zero. To this end we regard $\R^0$ as having as sole element the empty vector, the trivial topology, and the trivial vector space structure. Every such total function is linear, and this determines its adjoint.

We take such a function $f \type   \R^n \pr \R^m$,  where $n$ or $m$  is zero,  to be continuously differentiable if it is continuous, and we define 
$\myd(f)\type \R^n \dtimes \R^n \pr \R^m$ and $\myd^r(f)\type \R^n \dtimes \R^m \pr \R^n$ by:
\[\myd(f)(\bx, \by) \simeq \left\{ \begin{array}{ll} 0 & (f(\bx)\mydef)\\ 
                                                                          \myundef & (\mbox{otherwise})
                                                  \end{array}\right . 
\qquad
\myd^r(f)(\bx, \by) \simeq \left\{ \begin{array}{ll} 0 & (f(\bx)\mydef)\\ 
                                                                          \myundef & (\mbox{otherwise})
                                                  \end{array}\right .\]
%
The derivative $\myd(f)$ has domain $ \Dom(f) \times \R^n$ and is continuous (and so continuously differentiable) iff $f$ is, and a similar remark applies to $\myd^r(f)$. We understand  $\myd_{\bx}(f)$ and  $\myd^r_{\bx}(f)$ similarly to before. In case $f$ is linear (i.e., total) 
we have  $\myd_{\bx}(f) = f$ and $\myd^r_{\bx}(f) = f^{\dagger}$ as before and equation~(\ref{mode-rel}) relating the two modes continues to hold, as does equation~(\ref{f-from-r}) and also equations 
(\ref{double-reverse}) and (\ref{double-forward}) concerning derivatives with respect to two variables. In particular for $\mathrm{t}\type \R^n \rightarrow \R^0$ we have:
\[\begin{array}{lcl}(\myd_{\bx}\mathrm{t})\bx' & = &  \ast \qquad (\myd^r_{\bx} \mathrm{t})\ast = \mathbf{0}\end{array}\]

Regarding compositions we note  a useful fact:
\begin{fact} \label{fact} If $h\type \R^n \pr \R^l \; (l,n \geq 0)$ is constant on its domain, then it is continuously differentiable iff it is continuous and then, for any $\bx$ in its domain, $\myd_{\bx}(h)$ is the constantly $0$ function, as is $\myd^r_{\bx}(h)$.
\end{fact} 
It follows that the continuously differentiable functions, if taken in our wider sense, remain closed under composition and pointwise addition, and that the chain rule for forward derivatives continues to hold, as does that for reverse derivatives.  From now on whenever we consider partial functions from an $\R^n$ to an $\R^m$, we include the cases where $n$ or $m$ is $0$.

The projections 
$\pi_i^{m_0,\ldots, m_{k-1}}\type  \R^{m_0} \dtimes \ldots \dtimes\R^{m_{k-1}} \rightarrow \R^{m_i}$
%
are total linear functions, so we have:
\[\begin{array}{lcl}\myd_{\bx}\pi_i^{m_0,\ldots, m_{k-1}} & = & \pi_i^{m_0,\ldots, m_{k-1}}  \\[.3em]
(\myd^r_{\bx}\pi_i^{m_0,\ldots, m_{k-1}})\by & = & (0,\ldots,0,\by, 0,\ldots,0)\end{array}\]
%
Regarding   the tupling  
$\tuple{f_0,\ldots, f_{k-1}}\type  \R^n \pr  \R^{m_0} \dtimes \ldots \dtimes\R^{m_{k-1}} $ of  $f_i\type  \R^n \pr \R^{m_i}$
we have:
\[\begin{array}{lcl}\myd_{\bx} \tuple{f_0,\ldots, f_{k-1}} & = &   \tuple{\myd_{\bx}f_0,\ldots, \myd_{\bx}f_{k-1}} \\[0.3em]
 (\myd^r_{\bx} \tuple{f_0,\ldots, f_{k-1}})(\by_0,\ldots,\by_{k-1}) &  \simeq  & (\myd^r_{\bx}f_0)\by_0 + \ldots +  (\myd^r_{\bx}f_{k-1})\by_{k-1}   \end{array}\]
%
%
%

For the semantics of our  language we  work with infinitely differentiable functions, i.e., smooth ones. First we define smoothness classes $C^k$. 
We say that a  partial function $f \type  \R^n \pr \R^m$ is $C^0$ if it is continuous, and, inductively,  is $C^{k +1}$ if  $\myd(f)$ has domain $ \Dom(f) \times \R^n$ and is $C^{k}$.  This defines a decreasing sequence of classes of functions, and we say that $f$ is \emph{smooth} or $C^{\infty}$ if it is $C^k$ for all $k$. The $C^1$  functions are precisely the continuously differentiable ones. Using the chain rule for differentials one shows that the $C^k$ functions, and so too the smooth ones, are closed under composition. The projections are smooth, as are all linear functions and the $C^k$ functions, and so too the smooth ones, are closed under tupling. 




%

\subsection{Cppos of Differentiable Functions}

The subgraph partial order on partial functions between  powers of $\R$ interacts well with the differential structure:
%
\begin{proposition}\label{osupdiff}
\hspace{2em}
\begin{enumerate}
\item For any $f \leq g \type   \R^n\pr \R^m$ with open domain  we have:
\[\myd f = (\myd g)\res( \Dom(f) \times \R^n) \qquad \mydr f = (\mydr g)\res( \Dom(f) \times \R^m)\]
\item
For any compatible family of functions with open domain $f_i \type   \R^n \pr \R^m\; (i \in I)$, we have: 
\[\myd \bigvee_{i \in I} f_i = \bigvee_{i \in I} \myd  f_i \qquad  \qquad \mydr \bigvee_{i \in I} f_i = \bigvee_{i \in I} \mydr  f_i \]

\end{enumerate}
\end{proposition}
\ap{\begin{proof}For the first part,  if $m$ or $n$ is 0, the conclusion is immediate using Fact~\ref{fact}.
Otherwise, as $f$ and $g$ agree on an open set including $x$ we have: 
%
\[\jacobian_{\bx}(f)  \simeq  \jacobian_{\bx}(g)\]
and the conclusion follows.
For the second part, set $f = \bigvee f_i $. For the forward-mode derivative, using the first part we calculate:

\[\begin{array}{lclcl}
\myd f & = & \myd f \res (\Dom(f) \times \R^n)
           & = & \myd f \res \bigcup_{i \in I} (\Dom(f_i) \times \R^n)\\
           & = & \bigvee_{i \in I} \myd f \res (\Dom(f_i) \times \R^n)
           & = & \bigvee_{i \in I} \myd f_i\\
\end{array}\]
%
%
The proof for the reverse-mode  derivative is similar.

\end{proof}
}
\begin{proposition} \label{smoothdiff}
\hspace{2em}
\begin{enumerate}
\item Let $f \leq g \type   \R^n \pr \R^m$ be partial functions with open domain. Then $f$ is 
smooth 
if $g$ is.
\item Let $f_i \type   \R^n \pr \R^m$ be a compatible family of partial  functions with open domain, and with sup $f$. Then   $f$ is 
smooth 
if all the $f_i$ are. 
\end{enumerate}
\end{proposition}
\ap{\begin{proof}  For the first part, we prove by induction that if $g$ is $C^k$ then so is $f$.  For $k = 0$ we note that for any open $V \subseteq \R^m$, $f^{-1}(V) = g^{-1}(V) \cap \Dom(f)$.
For $k + 1$, as  $g$ is $C^{k+ 1}$, $\Dom(\myd g) = \Dom(g) \times \R^n$ and $\myd g$ is $C^k$. From part (1) of Proposition~\ref{osupdiff} we have $\myd f \leq \myd g$ and $\Dom(\myd f ) = \Dom(f) \times \R^n$. 
So $\myd f$ and $\myd g$ have open domain, $\myd f \leq \myd g$, and $\myd g$ is $C^k$. It follows from the induction hypothesis that $\myd f$ is $C^k$.
%
So $f$ is $C^{k + 1}$, as required.

For the second part we prove, by induction on $k$ that if all the $f_i$  are $C^k$, then so is $f$. For $k = 0$ this is clear. 
For $k+1$, we have, for all $i$, that $\Dom(\myd f_i) = \Dom(f_i) \times \R^n$ and $f_i$ is $C^k$.
From part (2) of Proposition~\ref{osupdiff} we have $\myd f = \bigvee_i \myd f_i$. So, first,
\[\begin{array}{lclcl}\Dom(\myd f) &  = &  \Dom(\bigvee_i \myd f_i) & = & \bigcup_i \Dom(\myd f_i)\\
                                                     &  = &  \bigcup_i  \Dom(f_i) \times \R^n & = & \Dom(f) \times \R^n\end{array}\]
%
%
and second, also using the induction hypothesis, we have that $\myd f$ is $C^k$. So $f$ is $C^{k + 1}$, as required.
\
\end{proof}
}
\POPLomit{ \begin{proof}  For the first part, we prove by induction that if $g$ is $C^k$ then so is $f$.  For $k = 0$ we note that for any open $V \subseteq \R^m$, $f^{-1}(V) = g^{-1}(V) \cap  \Dom(f)$.
For $k + 1$, as  $g$ is $C^{k+ 1}$, $ \Dom(\myd g) =  \Dom(g) \times \R^n$ and $\myd g$ is $C^k$. From part (1) of Proposition~\ref{osupdiff} we have $\myd f \leq \myd g$ and $ \Dom(\myd f ) =  \Dom(f) \times \R^n$. 
So $\myd f$ and $\myd g$ have open domain, $\myd f \leq \myd g$, and $\myd g$ is $C^k$. It follows from the induction hypothesis that $\myd f$ is $C^k$.
%
So $f$ is $C^{k + 1}$, as required.

For the second part we prove, by induction on $k$ that if all the $f_i$  are $C^k$, then so is $f$. For $k = 0$ this is clear. 
For $k+1$, we have, for all $i$, that $ \Dom(\myd f_i) =  \Dom(f_i) \times \R^n$ and $f_i$ is $C^k$.
From part (2) of Proposition~\ref{osupdiff} we have $\myd f = \bigvee_i \myd f_i$. So, first,
\[\begin{array}{lcl} \Dom(\myd f) & \hspace{-6pt} = & \hspace{-4pt}  \Dom(\bigvee_i \myd f_i) = \bigcup_i  \Dom(\myd f_i)\\
                                                  & \hspace{-6pt}  = & \hspace{-4pt} \bigcup_i   \Dom(f_i) \times \R^n =   \Dom(f) \times \R^n\end{array}\]
and second, additionally using the induction hypothesis, we have that $\myd f$ is $C^k$. So $f$ is $C^{k + 1}$, as required.
\end{proof}}

We write $\smooth{\R^n}{\R^m}$ for the coherent dcppo of smooth partial functions between $\R^n$ and $\R^m$.

\POPLsomit{\myacomment{write on connection with sheafs}}

%
%

 
\newcommand{\Cond}{\mathrm{Cond}}

\subsection{Conditionals and Recursion}

Differentiation and conditionals interact well.
The \emph{conditional combinator} $\Cond_{n,m}$ 
%
%
is defined for $p \type  R^n \pr \T$ and $f,g \type  R^n \pr \R^m$   by:
\[\Cond_{n,m}(p,f,g)(\bx)  \;\simeq\;   \left \{ \begin{array}{ll}
                                                                   f(\bx) & (p(\bx) = \dtrue)\\
                                                                   g(\bx) & (p(\bx) = \dfalse)\\
                                                                      \myundef & (\mbox{otherwise})
                                                              \end{array}\right .\]
%
where $\T= \{\dtrue,\dfalse\}$. The conditional combinator is  continuous. For differentiability, with $\T$  a discrete topological space, we have:
\begin{proposition} \label{cond-diff} Suppose $p$ is continuous (equivalently:  both $p^{-1}(\dtrue)$ and  $p^{-1}(\dfalse)$ are open). Then:
\[\begin{array}{lcl}\myd (\Cond_{n,m}(p,f,g)) & =  & \Cond_{(n+n),m}(p\circ \pi_1,\myd f, \myd g) \end{array}\]
 and 
 \[\begin{array}{lcl}\myd^r (\Cond_{n,m}(p,f,g)) & =  & \Cond_{(n+m),n}(p\circ \pi_1,\myd^r f, \myd^r g) \end{array}\]
 
%
Further, if  $f,g$ are smooth 
so is 
$\Cond_{n,m}(p,f,g)$.
%
%
\end{proposition}
\ap{\begin{proof}
Assume that $p$ is continuous. Set $h = \Cond_{n,m}(p,f,g)$. The domain of $h$ is open, indeed 
$\Dom(h) = (p^{-1}(\dtrue) \cap \Dom(f)) \cup (p^{-1}(\dfalse)  \cap \Dom(g))$
so $\myd h$ is defined. 

To prove the equality, choose $\bx \in \R^n$. There are three cases. First, if $p(\bx)\myundef $ then $h(\bx)\myundef$ and so $(\myd h)(\bx)\myundef$; the equality therefore holds at $\bx$.  Second if $p(\bx) = \dtrue $ then, as $h\res p^{-1}(\dtrue) = f\res p^{-1}(\dtrue)$ and $p^{-1}(\dtrue)$ is open, 
we see that, by part (1) of Proposition~\ref{osupdiff}, $\myd_{\bx}h =  \myd_{\bx} h\res p^{-1}(\dtrue) = \myd_{\bx} f \res p^{-1}(\dtrue) = \myd_{\bx}f $, 
and so the equality again holds at $\bx$. The third case is similar to the second.

Suppose further that $f,g$ are   smooth.  
We have $h\res p^{-1}(\dtrue) \leq f$, and so, by part (1) of Proposition~\ref{smoothdiff},  $h\res p^{-1}(\dtrue)$ is  smooth  as $f$ is and $h\res p^{-1}(\dtrue) $ has open domain.
Similarly $h\res p^{-1}(\dfalse)$ is  smooth. But then,  by part (2) of Proposition~\ref{smoothdiff}, $h$ is  smooth  as $h = h\res p^{-1}(\true) \vee h\res p^{-1}(\dfalse)$. 
\end{proof}}

In less formal terms than the proof, equality 
holds because if, say, the condition 
$p$ holds at $\bx \in \R^n$, 
it holds in a neighborhood $O$ of $\bx$.  
So $f$ and the conditional are equal throughout $O$, and  therefore have the same derivative there. 
The equality justifies the approaches to the differentiation of conditionals
described in the Introduction.

Differentiation and recursion also interact well.
For any continuous $f \type   P \times Q \rightarrow Q$ ($P$ a dcpo,  $Q$  a dcppo)  we write 
$\mu y \type  Q. f(x,y)$ for the least fixed-point (l.f.p.) of $f(x,-)$. 
It is the sup 
of the \emph{iterates} 
$\mu_n y\type  Q. f(x,y)$, defined inductively by:
\[\begin{array}{lcl}
 \mu_{n+1} y\type  Q. f(x,y) &  = &   f(x,\mu_n y\type  Q. f(x,y)) 
\end{array}\]
starting from  $\perp_Q$. As functions of $P$, the l.f.p. 
and the  iterates 
 are continuous. When $f\type Q \rightarrow Q$,  we  write $\mu y:Q.f(y)$, etc.
%
%

\begin{proposition}  \label{fix-diff}
\hspace{2em}
\begin{enumerate}
\item Set $Q = \smooth{\R^m}{\R^l}$ and $R =  \smooth{\R^m \dtimes \R^m}{\R^l}$.  Then if $F\type \R^n \times Q \rightarrow Q$  and $G\type \R^n \times  Q \times  R \rightarrow R$ 
are such that
%
%
%
\[\myd F(x,f) = G(x,f,\myd f) \quad (x \in \R^n, f \in Q)\]
we have:
\[\myd(\mu f \hspace{-1pt}\type \hspace{-1pt} Q.F(x,f)) \!=\! 
    \mu f' \hspace{-1pt}\type \hspace{-1pt} R. \, G(x,\mu f \hspace{-1pt}\type \hspace{-1pt} Q.F(x,f),f')\]
\item 

Set $Q = \smooth{\R^m}{\R^l}$ and $R =  \smooth{\R^m \dtimes \R^l}{\R^m}$. Then if $F\type \R^n \times Q \rightarrow Q$  and \\$G\type \R^n \times  Q \times  R \rightarrow R$ 
are such that
%
%
%
\[\myd^r F(x,f) = G(x,f,\myd^r f) \quad (x \in \R^n, f \in Q)\]
we have:
\[\myd^r(\mu f \hspace{-1pt}\type \hspace{-1pt} Q.F(x,f)) \!=\! 
    \mu f' \hspace{-1pt}\type \hspace{-1pt} R. \, G(x,\mu f \hspace{-1pt}\type \hspace{-1pt} Q.F(x,f),f')\]

\end{enumerate}

\end{proposition} 
\ap{\begin{proof} We only consider the forward-mode case as the reverse-mode case is similar. In one direction we prove by induction on $n$ that
 \[\myd(\mu_n f \hspace{-1pt}\type \hspace{-1pt} Q.F(x,f)) \!\leq\! 
    \mu f' \hspace{-1pt}\type \hspace{-1pt} R. \, G(x,\mu f \hspace{-1pt}\type \hspace{-1pt} Q.F(x,f),f')\]
    This is evident for $n = 0$. For $n+1$ we calculate (missing out  types):
    \[\begin{array}{lcl}
      \myd(\mu_{n+1} f .F(x,f)) 
        &  = &   \myd(F(x, \mu_{n} f .F(x,f))) \\
         &   = &  G(x,\mu_{n} f .F(x,f),\myd \mu_{n} f .F(x,f)) \\
          &  \leq &  G(x,\mu f .F(x,f),    \mu f'. \, G(x,\mu f. F(x,f),f')) \\
           &   =  &  \mu f'. \, G(x,\mu f. F(x,f),f')
      
    \end{array}\]
    
    In the other direction we prove by induction on $n$ that
 \[\mu_n f' \hspace{-1pt}\type \hspace{-1pt} R. \, G(x,\mu f \hspace{-1pt}\type \hspace{-1pt} Q.F(x,f),f')
\!\leq\! 
 \myd(\mu f \hspace{-1pt}\type \hspace{-1pt} Q.F(x,f)) \]
     This is evident for $n = 0$. For $n+1$ we calculate:
      \[\begin{array}{lcl}
       \mu_{n+1} f'. \, G(x,\mu f. F(x,f),f')
     &   =    &  G(x,\mu f .F(x,f),    \mu_n f'. \, G(x,\mu f. F(x,f),f')) \\
       &  \leq &  G(x,\mu f .F(x,f),\myd \mu f .F(x,f))\\
     &   = &  \myd F(x, \mu f .F(x,f))\\
     &   = &  \myd(\mu f .F(x,f)) \\

    \end{array}\]
   The conclusion then follows using the continuity of $\myd$ (part (2) of Proposition~\ref{osupdiff}).

\end{proof}}
\POPLomit{ \begin{proof}
We consider only the forward-mode case, as the reverse-mode case is similar. In one direction we prove by induction on $n$ that
 \[\myd(\mu_n f \hspace{-1pt}\type \hspace{-1pt} Q.F(x,f)) \!\leq\! 
    \mu f' \hspace{-1pt}\type \hspace{-1pt} R. \, G(x,\mu f \hspace{-1pt}\type \hspace{-1pt} Q.F(x,f),f')\]
    This is evident for $n = 0$. For $n+1$ we calculate (missing out  types):
    \[\begin{array}{l}
      \myd(\mu_{n+1} f .F(x,f)) \\
         =   \myd(F(x, \mu_{n} f .F(x,f))) \\
           =  G(x,\mu_{n} f .F(x,f),\myd \mu_{n} f .F(x,f)) \\
          \leq  G(x,\mu f .F(x,f),    \mu f'. \, G(x,\mu f. F(x,f),f')) \\
             =   \mu f'. \, G(x,\mu f. F(x,f),f')
      
    \end{array}\]
    
    In the other direction we prove by induction on $n$ that
 \[\mu_n f' \hspace{-1pt}\type \hspace{-1pt} R. \, G(x,\mu f \hspace{-1pt}\type \hspace{-1pt} Q.F(x,f),f')
\!\leq\! 
 \myd(\mu f \hspace{-1pt}\type \hspace{-1pt} Q.F(x,f)) \]
     This is evident for $n = 0$. For $n+1$ we calculate:
      \[\begin{array}{l}xq
       \mu_{n+1} f'. \, G(x,\mu f. F(x,f),f')\\
       =     G(x,\mu f .F(x,f),    \mu_n f'. \, G(x,\mu f. F(x,f),f')) \\
        \leq  G(x,\mu f .F(x,f),\myd \mu f .F(x,f))\\
       =\myd F(x, \mu f .F(x,f))\\
       = \myd(\mu f .F(x,f)) \\

    \end{array}\]
The conclusion then follows using the continuity of $\myd$ (Part 2 of Proposition~\ref{osupdiff}).
\end{proof}}

Although we do not do so here, this proposition can be used to justify  code transformations of recursive function definitions.



\section{Denotational semantics}\label{densem}


We begin with a denotational semantics of types as  powers of the reals:
\[\begin{array}{lcl}
\den{\real} & = & \R\\
\den{\unit} & = & \R^0\\
\den{T \times U} & = & \den{T} \dtimes \den{U}\\
\end{array}\]
Note that $\den{T} = \R^{|T|}$, where,  $|\real| = 1$, $|\unit| = 0$ and, recursively,  $|T \times U| = |T| + |U|$.
Then, for the semantics of environments $\G = x_0\type T_0, \ldots, x_{n-1} \type T_{n-1}$ 
we set 
\[\den{\G} = \den{T_0} \dtimes \ldots \dtimes \den{T_{n-1}} \]
%
%
and for that of function environments $\Phi =  f_0\type T_0\! \rightarrow \! U _n, \ldots, f_{n-1}\type T_{n-1}\! \rightarrow\! U_{n-1} $  we set 
 \[\den{\Phi} =   \smooth{\den{T_0}}{\den{U_0}}\times  \ldots \times \smooth{\den{T_{n-1}}}{\den{U_{n-1}}} \]
%
%
We use $\gamma,\delta$ to range over $\den{\G}$ and  $\phi$ over $\den{\Phi}$. 
The vectors $\gamma \in \den{\G}$ 
correspond to  semantic environment functions $\rho$ on $\Dom(\G)$ with $\rho(x_i) = \pi_i(\gamma)$, and this correspondence is 1-1. We take advantage of it to write $\gamma[a/x]$, when $a \in \den{T}$,  for an element of $\den{\G[x:T]}$ (assuming $\G$ and $T$  available from the context). We understand $\phi[\alpha/f]$ similarly, for $\phi \in \den{\Phi}$ and 
$\alpha \in \smooth{\den{T}}{\den{U}}$.





The denotational semantics of a term $\Phi \!\mid \!\G \! \vdash \! M\type T$ will be a continuous function:
\[\den{\Phi}  \xrightarrow{\den{\Phi \mid \G \vdash M:T}} \smooth{\den{\G}}{\den{T}}\]
and that of a boolean term $\Phi\! \mid \! \G \!\vdash \! B$ will be continuous functions:
\[\den{\Phi}  \xrightarrow{\den{\Phi \mid \G \vdash B}} \cont{\den{\G}}{\T}\]
%
When the environments and types are understood from the context, we may just write $\den{M}$ or $\den{B}$.

For the semantics of operation and predicate symbols, for every $\op \type  T\rightarrow U$ we assume available a smooth function $\den{\op}: \den{T} \pr \den{U}$ and for every  $\pred\type  T$  we assume available a a continuous function $\den{\pred}: \den{T} \pr \T$, such
 that, for every closed value $V\type T$:
%
\begin{equation} \label{req1} \den{\op}(\den{V})  \simeq  \den{\ev(\op,V)}\end{equation}
 and, for every closed value $V\type T$:
\begin{equation} \label{req2} \den{\pred}(\den{V})  \simeq  \den{\bev(\pred,V)}\end{equation}
%


\begin{figure}[h]

\[ {
\begin{array}{lcll} \centering
\den{x}\phi \s \gamma   &\hspace{-8pt} \simeq   &\hspace{-3pt} \gamma(x)\\[12pt]
  \den{r}\phi \s \gamma   &\hspace{-8pt} \simeq   &\hspace{-5pt} r \qquad (r \in \R)\\[12pt] 
  \den{M + N}\phi \s \gamma   &\hspace{-8pt} \simeq   &\hspace{-7pt}  \den{M}\phi \s \gamma    +  \den{N}\phi \s \gamma   \\[12pt] 
    \den{\op(M)}\phi \s \gamma   &\hspace{-8pt} \simeq   &\hspace{-5pt}   \den{\op}(\den{M}\phi \s \gamma )\\[12pt] 
  \den{\mylet  x\type T  \hspace{-3pt} \mybe \hspace{-3pt} M \hspace{-2pt} \myin \hspace{-2pt} N}\phi \s \gamma   &\hspace{-8pt} \simeq   &\hspace{-5pt} \den{N} \phi \s (\gamma[\den{M}\phi \s \gamma /x])\\[12pt]
    \den{\ast}\phi \s \gamma   &\hspace{-8pt} \simeq   &\hspace{-5pt} \ast\\[12pt] 
   \den{\tuple{M, N}_{T,U}}\phi \s \gamma   &\hspace{-8pt} \simeq   &\hspace{-5pt} 
       \tuple{\den{M}\phi \s \gamma ,\den{N}\phi \s \gamma } \\[12pt] 
         \den{\ttfst_{T,U}(M)}\phi \s \gamma   &\hspace{-8pt} \simeq   &\hspace{-5pt}   \den{\fst}(\den{M}\phi \s \gamma )\\[12pt] 
     \den{\ttsnd_{T,U}(M)}\phi \s \gamma   &\hspace{-8pt} \simeq   &\hspace{-5pt}   \den{\snd}(\den{M}\phi \s \gamma )\\[12pt] 

         \den{\myif \! B \! \mythen \! M \! \myelse \! N}\phi \s \gamma  &\hspace{-8pt} \simeq   &\hspace{-5pt}
   \left \{ \hspace{-5pt}\begin{array}{cl}
 \den{M}\phi \s \gamma  & ( \den{B}\phi \s \gamma   \simeq  \dtrue)\\
 \den{N}\phi \s \gamma  & ( \den{B}\phi \s \gamma   \simeq  \dfalse)\\
 \myundef & (\mbox{otherwise})
  \end{array}
  \right . \\[16pt] 
 \begin{array}{l}\hspace{-5pt}[\![\hspace{-2.5pt}\myletrec f(x\type T)\type U 
  \mybe M \myin N]\!]\phi \s \gamma  \end{array} &\hspace{-8pt} \simeq   &\hspace{-5pt} 
\hspace{-5pt}\begin{array}{l}\den{N}(\phi[\mu \alpha\type  \smooth{\den{T}}{\!\den{U}}.\,\\
\hspace{32pt}  \lambda a\type  \den{T}.\,\den{M}(\phi[\alpha/f])a/f]) \s \gamma\end{array}\\[12pt] 
 \den{f(M)}\phi \s \gamma  &\hspace{-8pt} \simeq   &\hspace{-5pt} \phi(f)(\den{M}\phi \s \gamma )\\[12pt] 
  \den{\dRD{x\type T.\,N}{ L}{M}}\phi \s \gamma  &\hspace{-8pt} \simeq   &\hspace{-5pt}
 \hspace{-5pt}\begin{array}{l} \mydr_{\den{L}\phi \s \gamma }(a \in \den{T} \mapsto \den{N}\phi \s (\gamma[a/x]))
 (\den{M}\phi \s \gamma )\end{array}\\[12pt]
       \den{\true}\phi \s \gamma   &\hspace{-8pt} \simeq   &\hspace{-5pt} \dtrue\\[12pt] 
     \den{\false}\phi \s \gamma   &\hspace{-8pt} \simeq   &\hspace{-5pt} \dfalse\\[12pt] 
    \den{\pred(M)}\phi \s \gamma   &\hspace{-8pt} \simeq   &\hspace{-5pt} \den{\pred}(\den{M}\phi \s \gamma )\\\ 
\end{array}}\]
\vspace{-10pt}
\caption{Denotational semantics}
\label{den-sems}
\end{figure}

The denotational semantics is given in Figure~\ref{den-sems}. Note that it uses  the no-free-variable assumption in the clause for recursive functions. Apart from the semantics of reverse differentiation, which uses the reverse-mode derivative $\myd^r$, it is quite standard. However, the facts that the denotations of terms  carry smooth functions to smooth functions, and that the denotations of boolean terms carry smooth functions to continuous ones, use the mathematics developed in the previous section, particularly: the chain rule, e.g., for function application and let constructs; the preservation of smooth functions by the conditional combinator; the remarks on products; and, for recursive function definitions, the fact that the lub of an increasing sequence of smooth functions is smooth.

If a term $M$ contains no function variables (or variables), $\den{M}\phi \s \gamma$ is independent of $\phi$ (and $\gamma$), and we  write 
$\den{M}\gamma$ (resp., $\den{M}$) for it. Trace terms $C$ have no function variables,  and  closed values $V$ have no function variables or variables. 



Using the semantics of terms, we can define the denotational semantics of value environments, function environments, and closures, the latter two by a mutual structural induction.

\begin{itemize}
\item
For every $\rho\type \G$, where $\G = x_0\type T_0, \ldots, x_{n-1} \type T_{n-1}$,
we define $\den{\rho\type\G} \in \den{\G}$ by:
\[\den{\rho\type\G} = (\den{\rho(x_0)},\ldots,\den{\rho(x_{n-1})})\]
%
\item
For every $\varphi\type \Phi$, where
$\Phi =  f_0\type T_0\! \rightarrow \! U _n, \ldots, f_{n-1}\type T_{n-1}\! \rightarrow\! U_{n-1} $,
 we define $\den{\varphi \type \Phi} \in \den{\Phi}$ by:
\[\den{\varphi \type \Phi} = (\den{\varphi(f_0)},\ldots,\den{\varphi(f_{n-1})})\]
%
\item For every $\Cl = \fclo{\varphi}{f(x \type T)\type  U.\, M}\type T \!\rightarrow\! U $ we define $\den{\Cl} \in \smooth{\den{T}}{\den {U}}$
by:
\[\begin{array}{lcl}\den{\Cl} & =  & \mu \alpha\type \smooth{\den{T}}{\den{U}}.\, \lambda a\type \den{T}.\den{M}(\den{\varphi\type \Phi_{\varphi}}[\alpha/f]) \s a
\end{array}\]
%
%
\end{itemize}
We  omit $\G$ and $\Phi$ from $\den{\rho\type\G}$ and $\den{\varphi \type \Phi} $ when they can be understood from the context.
\begin{lemma}  \label{val-den} For any type $T$ we have:
\begin{enumerate}
\item The denotation $\den{V}$ of any value $V:T$ exists.
\item For any $v \in \den{T}$ there is a unique value $V:T$ such that $v = \den{V}$.
\end{enumerate}
\end{lemma}

\mycomment{Martin: the following two lemmas could be omitted if needed. Maybe they should anyway be omitted as they don't show up below}

The following two fairly standard results about evaluation contexts are needed to show adequacy.
\begin{lemma}[Evaluation context compositionality] \label{ctxt-sem1} 
Suppose $\Phi\!\mid\! \G \vdash M\type U$, $\Phi\!\mid\! \G \vdash E[M]\type T$, and $\Phi\!\mid\! \G \vdash N\type U$. Then:
\[\den{M}\phi \s \gamma  \simeq  \den{N}\phi \s \gamma \implies \den{E[M]}\phi \s \gamma  \simeq  \den{E[N]}\phi \s \gamma \qquad (\phi\in \den{\Phi}, \gamma \in \den{\G})\]

Analogous results hold for: boolean terms in evaluation contexts, $E[B]$; terms in boolean evaluation contexts $E_{\bool}[M]$; and boolean terms in boolean evaluation contexts $E_{\bool}[B]$.
\POPLomit{
\hspace{2em}
\begin{enumerate}
\item Suppose $\Phi\!\mid\! \G \vdash M\type U$, $\Phi\!\mid\! \G \vdash E[M]\type T$, and $\Phi\!\mid\! \G \vdash N\type U$. Then:
\[\den{M}\phi \s \gamma  \simeq  \den{N}\phi \s \gamma \implies \den{E[M]}\phi \s \gamma  \simeq  \den{E[N]}\phi \s \gamma \qquad (\phi\in \den{\Phi}, \gamma \in \den{\G})\]
\item  Suppose $\Phi\!\mid\! \G \vdash B$, $\Phi\!\mid\! \G \vdash E[B]\type T$, and $\Phi\!\mid\! \G \vdash B'$. Then:
\[\den{B}\phi \s \gamma  \simeq  \den{B'}\phi \s \gamma \implies \den{E[B]}\phi \s \gamma  \simeq  \den{E[B']}\phi \s \gamma \qquad (\phi\in \den{\Phi}, \gamma \in \den{\G})\]

\item  Suppose $\Phi\!\mid\! \G \vdash M\type U$, $\Phi\!\mid\! \G \vdash E_{\bool}[M]$, and $\Phi\!\mid\! \G \vdash N\type U$. Then:
\[\den{M}\phi \s \gamma  \simeq  \den{N}\phi \s \gamma \implies \den{E_{\bool}[M]}\phi \s \gamma  \simeq  \den{E_{\bool}[N]}\phi \s \gamma \qquad (\phi\in \den{\Phi}, \gamma \in \den{\G})\]
\item  Suppose $\Phi\!\mid\! \G \vdash B$, $\Phi\!\mid\! \G \vdash E_{\bool}[B]$, and $\Phi\!\mid\! \G \vdash B'$. Then:
\[\den{B}\phi \s \gamma  \simeq  \den{B'}\phi \s \gamma \implies \den{E_{\bool}[B]}\phi \s \gamma  \simeq  \den{E_{\bool}[B']}\phi \s \gamma \qquad (\phi\in \den{\Phi}, \gamma \in \den{\G})\]
\end{enumerate}
}
\end{lemma}
\POPLomit{ \begin{proof}
\end{proof}}


\begin{lemma}[Evaluation context strictness] \label{ctxt-sem2}
\hspace{2em}
\begin{enumerate}
\item Suppose that $\Phi\!\mid\! \G \vdash E[M]\type T$ and $\Phi\!\mid\! \G \vdash M\type U$. Then, for any \ $x \notin \FV(E)$, we have:
\[\den{E[M]}\phi \s \gamma \mydef \implies \den{M}\phi \s \gamma\mydef \quad (\phi\in \den{\Phi}, \gamma \in \den{\G})\]
and so:
\[\den{E[M]}\phi \s \gamma \simeq \den{E[x]}\phi \s (\gamma[\den{M}\phi \s \gamma/x]) \quad (\phi\in \den{\Phi}, \gamma \in \den{\G})\]
The analogous result holds for terms in boolean contexts  $E_{\bool}[M]$.
\POPLomit{
\item Suppose that $\Phi\!\mid\! \G \vdash E_{\bool}[M]$ and $\Phi\!\mid\! \G \vdash M\type U$. Then, for any \ $x \notin \FV(E_{\bool})$, we have:
\[\den{E_{\bool}[M]}\phi \s \gamma \mydef \implies \den{M}\phi \s \gamma\mydef \quad (\phi\in \den{\Phi}, \gamma \in \den{\G})\]
and so:
\[\den{E_{\bool}[M]}\phi \s \gamma \simeq \den{E_{\bool}[x]}\phi \s \gamma[\den{M}\phi \s \gamma/x] \quad (\phi\in \den{\Phi}, \gamma \in \den{\G})\]
}
\item Suppose that $\Phi\!\mid\! \G \vdash E[B]\type T$ and $\Phi\!\mid\! \G \vdash B$. Then:
\[\den{E[B]}\phi \s \gamma \mydef \implies \den{B}\phi \s \gamma\mydef \quad (\phi\in \den{\Phi}, \gamma \in \den{\G})\]
The analogous result holds for boolean terms in boolean contexts  $E_{\bool}[B]$.

\POPLomit{
\item Suppose that $\Phi\!\mid\! \G \vdash E_{\bool}[B]$ and $\Phi\!\mid\! \G \vdash B$. Then:
\[\den{E_{\bool}[B]}\phi \s \gamma \mydef \implies \den{B}\phi \s \gamma\mydef \quad (\phi\in \den{\Phi}, \gamma \in \den{\G})\]
}

\end{enumerate}

\end{lemma}

\section{Adequacy}\label{adequacy}

We present 
our main 
results,  on the correspondence between operational and denotational semantics. Taken together these are our adequacy theorems. As well as the usual correctness and completeness theorems, there are results peculiar to differentiation, both of interest in themselves and also necessary for the others.
Theorem~\ref{for_diff_corr} shows the correctness of our source code transformation of trace terms, in that formal differentiation corresponds to actual differentiation.
Theorem~\ref{op-corr} has two parts. The first is the usual statement of correctness for the ordinary evaluation relation. The second is a statement of correctness of symbolic evaluation. This states that the trace term resulting from symbolic evaluation of a term has the same denotation not only at the environment used for the symbolic evaluation but on a whole open set including it.  This, and Theorem~\ref{for_diff_corr}, are needed to prove the ordinary correctness of the ordinary evaluation relation in the case of differentiation as in that case ordinary evaluation proceeds by first symbolically differentiating and then applying our source code transformation. 
Finally Theorem~\ref{op-com} is the expected completeness theorem, that if the semantics of a term is defined, then both its ordinary and symbolic evaluation terminate. Its proof makes use of Proposition~\ref{ord-symb-con} which interpolates a symbolic evaluation inside any ordinary (non-boolean) evaluation.

\subsection{Operational Correctness}

\begin{theorem}[Reverse-mode differentiation] \label{for_diff_corr}
Suppose that $  \G[x\type T] \vdash C\type U$, $ \G\vdash V\type T$ and  $ \G\vdash W\type U$ (and so $\G \vdash \dRD{ x\type T.\,C}{ V}{W}\type T$).
Then, for any $\gamma \in \den{\G}$, we have:  
\[ \den{\newcalRprop{ x \type T}{C}{V}{W}{}}\gamma \simeq \den{\dRD{x\type T.\,C}{ V}{W}}\gamma   \]
\end{theorem}
\ap{\begin{proof} 
The proof is by structural induction on $C$. We give a representative case. Suppose  $C$ is $ \mylet   y\type U  \mybe D \myin E$. 
Set $\gamma_a = \gamma[a/x] \in \den{\G[x\type T]}$, for any $a \in \den{T}$, $\gamma_V = \gamma_{\den{V}\gamma}$, and  
$\gamma' = \gamma_V[\den{D}\gamma_V/y] \in \den{\G[x\type T][y \type U]}$. 
 Then
\[ \den{\newcalRfig{ x\type T } {C}{V}{W}}\gamma  \simeq  
\den{\newcalRfig{   x\type T  }{E} { V}{W}   +_T  (\mylet  \ov{y}\type  S\! \mybe\! \newcalRfig{ y\type S }{E} { y}{W} \! \myin\! \, \newcalRfig{   x\type T }{D}{V}{\ov{y}})
}\gamma'\]
with $x \notin \FV(V,W)$, $y \notin \FV(W)$, and $y,\ov{y}\notin \FV(V,D)$.
%
  %
We calculate: 
\[\hspace{-0pt}\begin{array}{lcl}

   \mydr_{\den{V} \gamma }(a \in \den{T} \mapsto \den{ \mylet   y\type U  \mybe D \myin E} \gamma_a) \den{W} \gamma \\
   
  \simeq \quad  \mydr_{\den{V} \gamma }(a \in \den{T} \mapsto  
                                                  \den{E}\gamma_a[\den{D}\gamma_a/y]  )\den{W} \gamma \\

 \simeq \quad  \mydr_{\den{V} \gamma }
  ((c \in \den{T \times U} \mapsto \den{E}\gamma_{\pi_0(c)}[\pi_1(c)/y]) 
   \; \circ 
   \tuple{a \in \den{T} \mapsto a , \; a \in \den{T} \mapsto \den{D}\gamma_a})\den{W} \gamma\\

 \simeq \quad  \mydr_{\den{V}\gamma}(\tuple{a \in \den{T} \mapsto a , \; a \in \den{T} \mapsto \den{D}\gamma_a})
   [ \\ \hspace{60pt}(\mydr_{\tuple{\den{V} \gamma , \den{D}\gamma_V}}
  (c \in \den{T \times U} \mapsto \den{E}\gamma_{\pi_0(c)}[\pi_1(c)/y]) )
\den{W} \gamma]\\

 \simeq \quad  \mydr_{\den{V}\gamma}(\tuple{a \in \den{T} \mapsto a , \; a \in \den{T} \mapsto \den{D}\gamma_a})
   [ \\ \hspace{54pt}\tuple{\mydr_{\den{V} \gamma}
                                      (a \in \den{T} \mapsto \den{E}\gamma_{a}[\den{D}\gamma_V/y]) , 
   \mydr_{ \den{D}\gamma_V}
                                     (b \in \den{U} \mapsto \den{E}\gamma_{\den{V} \gamma}[b/y])}
\den{W} \gamma]\\
   
 \simeq \quad  \mydr_{\den{V}\gamma}(\tuple{a \in \den{T} \mapsto a , \; a \in \den{T} \mapsto \den{D}\gamma_a})
   [ \\ \hspace{30pt}\tuple{\mydr_{\den{V} \gamma}
                                      (a \in \den{T} \mapsto \den{E}\gamma_{a}[\den{D}\gamma_V/y]) \den{W} \gamma,  \mydr_{ \den{D}\gamma_V}
                                     (b \in \den{U} \mapsto \den{E}\gamma_{\den{V} \gamma}[b/y])\den{W} \gamma}]\\

  \simeq \quad \mydr_{\den{V}\gamma'}(a \in \den{T} \mapsto a)  
                          [\mydr_{\den{V} \gamma}  (a \in \den{T} \mapsto \den{E}\gamma_{a}[\den{D}\gamma_V/y]) \den{W} \gamma]\\
     \hspace{35pt}  +  \; \mydr_{\den{V}\gamma'}(a \in \den{T} \mapsto \den{D}\gamma') 
                            [\mydr_{ \den{D}\gamma_V} (b \in \den{U} \mapsto \den{E}\gamma_{\den{V} \gamma}[b/y])\den{W} \gamma ]\\

 \simeq \quad  \mydr_{\den{V} \gamma}  (a \in \den{T} \mapsto \den{E}\gamma_{a}[\den{D}\gamma_V/y]) \den{W} \gamma\\
     \hspace{35pt}  +  \; \mydr_{\den{V}\gamma'}(a \in \den{T} \mapsto \den{D}\gamma') 
                            [\mydr_{ \den{D}\gamma_V} (b \in \den{U} \mapsto \den{E}\gamma_{\den{V} \gamma}[b/y])\den{W} \gamma ]
                            \\



 
 



\simeq   \quad \den{\newcalRfig{   x\type T  }{E} { V}{W}   +_T  (\mylet  \ov{y}\type  S\! \mybe\! \newcalRfig{ y\type S }{E} { y}{W} \! \myin\! \, \newcalRfig{   x\type T }{D}{V}{\ov{y}})
}\gamma'

\end{array}\]
Note the use of Equation (\ref{double-reverse}) in the the fourth step.
\end{proof}}
          \begin{theorem}[Operational correctness] \label{op-corr}
Suppose  that $\Phi\! \mid\! \G \! \vdash\!  M\! \type\! T$, $\vdash \varphi \type \Phi$, and $\vdash  \rho : \G$. Then:
\begin{enumerate}
\item \textbf{Operational semantics.}  
\[\varphi \mid \rho \vdash M \impe V \implies \den{M}\den{\varphi}\den{\rho} =  \den{V}\]
 (and similarly for boolean terms). 
\item 
\textbf{Symbolic operational semantics.} 
%
%
\[\begin{array}{l}\varphi \mid \rho \vdash M \imps C \implies 
\exists \, O \subseteq_{\mbox{open}} \den{\G}.\, 
\den{\rho} \in O \wedge
\forall \gamma \in O.\, \den{M}\den{\varphi}\gamma  \simeq  \den{C}\gamma
 \end{array}\] 
%

%
\end{enumerate}
 \end{theorem}
 \ap{\begin{proof} The two parts are proved by mutual induction on the size of the proofs that establish the given operational relations, and by cases on the form of $M$. As an example case of the second part, suppose $M$ has the form $ \mylet   x:T  \mybe  V \myin L$. Then for some  $D$ and $V'$ we have a smaller proof of
 $ \varphi \mid \rho[V'/x]  \vdash L \imps D$, where $V' = \rho(V)$,
and $C$ has the form $\mylet   x:T  \mybe V \myin D$.

By the induction hypothesis there is an open set $O$ such that $\den{\rho[V'/x]} \in O$ and, for all $\delta \in O$, we have $\den{L}\den{\varphi}\delta  \simeq  \den{D}\delta$.
Set  $\theta = \gamma \in \den{\Gamma} \mapsto \gamma[\den{V}\gamma/x]$.
As $\theta$ is continuous, $O' \eqdef \theta^{-1}(O)$ is open. We show it is the required open set. 
\begin{itemize}
\item[-] First, $\den{\rho} \in O'$ as we have: 
$\theta(\den{\rho})  =   \den{\rho}[\den{V}\den{\rho}/x]   =  \den{\rho}[\den{V'}/x]   =  \den{\rho[V'/x]}  \in   O$.
\item[-]  Second, for any $\gamma \in O'$ we have:
\[\begin{array}{lcll}
\den{\mylet   x:T  \mybe  V \myin L}\den{\varphi}{\gamma} 
& \simeq &  \den{L}\den{\varphi}\gamma[\den{V}\gamma/x] \\ 
& \simeq  & \den{D}\gamma[\den{V}\gamma/x] & (\mbox{as } \theta(\gamma) \in O)\\
  &\simeq &  \den{\mylet   x:T  \mybe  V \myin D}{\gamma} \\
\end{array}\]
\end{itemize}
 \end{proof}}
The following corollary shows that our strategy of first symbolically reducing to produce a trace term,  then symbolically differentiating, is correct.

\mycomment{Martin: note the remark on our proof of this theorem.}

\begin{corollary} \label{M-cor}

Suppose that $ \Phi\mid \G[x\type T] \vdash M \type U$, $\G\vdash V\type T$, and $ \G\vdash W\type U$.
Then, for any $\vdash \varphi : \Phi$ and  $\vdash  \rho : \G$ we have:
\[\begin{array}{l} 
   \varphi \!\mid\! \rho[V\rho/x] \vdash  M \imps C \implies 
   \den{\dRD{x:T.\,M}{V}{W}}\den{\varphi}\den{\rho}  \simeq   \den{\dRD{x:T.\,C}{ V }{ W}}\den{\varphi}\den{\rho}
   \end{array} \]
  \end{corollary}
\POPLomit{ \begin{proof} By part (2) of Theorem~\ref{op-corr} there is an $O \subseteq_{\mbox{open}} \den{\G}$ such that $\den{\rho[V\rho/x] } \in O$ and for all $\gamma \in O$, we have 
$ \den{M}\den{\varphi}\gamma  \simeq  \den{C}\gamma$.
Let  $O' \subseteq \den{T}$ be the $b \in \den{T}$ such that $\den{\rho}[b/x] \in O$. This defines an open set containing $\den{V}\den{\rho}$.
 
 Since $ \den{M}\den{\varphi}\gamma  \simeq  \den{C}\gamma$ for all $\gamma \in O$, the functions
\[a \in \den{T} \mapsto  \den{M}\den{\varphi}\den{\rho}[a/x]\] 
and 
\[a \in \den{T} \mapsto  \den{C}\den{\rho}[a/x]\]
are equal on the open set $O'$. So, as that set contains $\den{V}\den{\rho}$ the functions have the same reverse-mode derivative there, and the conclusion follows using the semantics 
of reverse-mode  differentiation.
\end{proof}}

\subsection{Operational Completeness}

 
We turn to proving operational completeness, that the evaluation, or symbolic evaluation, of a term terminates when it should, i.e., when its denotation is defined. For terms $M$ write 
 $\varphi \mid \rho \vdash M \myDef$ when, for some $V$,
 $\varphi \mid \rho \vdash M \Rightarrow V$  and, assuming $\varphi$ and $\rho$ known from the context, say that $M$ \emph{terminates}; we adopt similar  terminology  for boolean terms.

To prove operational completeness 
we use a standard  strategy: first proving operational completeness for an auxiliary ``approximation language'' in which recursive definitions are replaced by approximations to them, which we call \emph{limited recursive definitions}, and then lifting that result to the main language. Specifically we replace the syntactic form 
for recursive definitions of the main language
%
%
 by the family of syntactic forms:
\[\amyletrec{n} f(x: T):U \mybe M \myin N \qquad (n \in \N)\]
with the evident typing rule,
and make analogous consequential changes in the various definitions. 
%
%
%

In the definition of  function environments and closures, the clause for closures becomes:
\begin{itemize}
\item[-]  If $\FFV(M)\backslash \{f\} \subseteq  \Dom(\varphi)$, $\FV(M) \subseteq \{x\}$, and $n \in \N$, 
  then \[\tuple{n,\varphi, f,x,T,U,M} \] is a closure, written as: $\afclo{n}{\varphi}{f(x:T):  U.\, M}$.
\end{itemize}
The limited recursive definition redexes are 
$\!\!\amyletrec{n} f(x \type  T) \type  U \mybe M \myin N$;
%
their evaluation rules are  in Figure~\ref{osbr}.  There, and below, we write $\ldash$ for limited recursion language judgements.

 \begin{figure}[h]
  \vspace{7.5pt}
   \[\hspace{-5.5pt}\begin{array}{l}\frac{\varphi[\afclo{n}{\varphi}{f(x:T):  U.\, M}/f] \mid \rho \ldash N \impe V' }
                     {\varphi \mid \rho \ldash \amyletrec{n} f(x: T):U \mybe M \myin N \impe  V'}
            \\\\

           \frac{\varphi \mid \rho \ldash V \impe V' \;\;
            \varphi'[\afclo{n}{\varphi'}{f(x:T):  U.\, M}/f]\mid \{x \mapsto V'\} \ldash M \impe W}
                     {\varphi \mid \rho \ldash f(V) \impe  W}
            
        \quad  \ms{ \mbox{where   $ \varphi(f) = \afclo{n+1}{\varphi'}{f(x:T):  U.\, M}$}}\\\\
 \frac{\varphi[\afclo{n}{\varphi}{f(x:T):  U.\, M}/f] \mid \rho \ldash N \imps C }
                     {\varphi \mid \rho \ldash \amyletrec{n} f(x: T):U \mybe M \myin N \imps  C}
            \\\\

              \frac{\varphi \mid \rho \ldash V \impe V' \;\;
                       \varphi'[ \afclo{n}{\varphi'}{f(x:T):  U.\, M}/f]\mid \{x \mapsto V' \} \ldash M \imps C}
                     {\varphi \mid \rho \ldash f(V) \imps   \mylet x \type T \mybe V \myin C}
        \quad  \ms{ \mbox{where   $ \varphi(f) = \afclo{n+1}{\varphi'}{f(x:T):  U.\, M}$}}
            \end{array}\]
              \vspace{7.5pt}
\caption{Operational semantics of bounded recursion}
\label{osbr}
\end{figure}

 For the closure typing judgement $\Cl\type  T\rightarrow U$ 
 we substitute:
 \[\frac
 {\ldash \varphi: \Phi  \quad \Phi[f\type T \rightarrow U] \mid  x \type T \ldash M\type  U}
 {\ldash  \afclo{n}{\varphi}{f(x:T):  U.\, M}\,\type \, T \rightarrow U}\]
%




As regards  the denotational semantics,  
the clause for limited recursive definitions is:
\[  \begin{array}{l} \den{\amyletrec{n} f(x \type  T) \type U \mybe M \myin N}(\phi)(\gamma) \;\;   \simeq  \\
\hspace{80pt}  \den{N}(\phi[(\mu_n \alpha \type  \smooth{\den{T}}{\den{U}}.\, \lambda a \type  \den{T}.\, \den{M}(\phi[\alpha/f])\s a)/f]) \s (\gamma)
\end{array}\]
The results for the operational and denotational semantics carry over to the restricted setting, and we refer to them in the same way as we do to the unrestricted versions. 

\mylcomment{as above had wrong bracketing, check if this persists to proofs}

Relating the two languages, the \emph{$n$-th approximant $M^{(n)}$} of a  language term $M$ 
is obtained by replacing every recursive definition  in $M$ by an $n$-limited one (similarly for boolean terms) and 
$n$-th approximants of closures and function  environments are defined by structural recursion:
\[\begin{array}{lcl}\fclo{\varphi}{f(x\type T)\type   U.\, M}^{(n)} & =  & \afclo{n}{\varphi^{(n)}}{f(x\type T)\type   U.\, M^{(n)}}\end{array}\]
\[\{\ldots, f_i \mapsto \Cl_i, \ldots \}^{(n)} = \{\ldots, f_i \mapsto \Cl^{(n)}_i, \ldots\}\]


The terms $M^{(n)}$ can be defined by structural recursion; we just give one clause of the definition:
\[\begin{array}{l}(\!\myletrec f(x: T):U \mybe M \myin N)^{(n)} \;\; = \;\; 
\amyletrec{n} f(x: T):U \mybe M^{(n)} \myin N^{(n)}\end{array}\]
%

Approximation preserves typing judgments. 

Termination is proved for the approximation language  by structural induction via a suitable notion of \emph{computability}.
\begin{itemize}[leftmargin = 0.25cm]
\item[-]   A  closure \[\ldash  \afclo{n}{\varphi}{f(x:T):  U.\, M}: T \rightarrow U\] is computable iff 
$n = 0$ or $n > 0$ and, for all $\ldash V:T$ we have:
%

\[\begin{array}{l} \hspace{-6pt}\den{M}\den{(\varphi [\afclo{n-1}{\varphi}{f(x\type T)\type  U.\, M}/f]) \type \Phi_{\varphi}[f\type T \rightarrow U]}\den{V}\!\mydef  
  \implies  \\
 \hspace{40pt} \varphi[\afclo{n-1}{\varphi}{f(x\type  T)\type   U.\, M}/f] \!\mid\!\{ x \!\mapsto\!  V\} \ldash M \!\myDef\end{array}\]

 %
\item[-] A function environment $\ldash  \varphi:\Phi$ is computable iff $\ldash  \varphi(f): \Phi(f)$ is a computable closure, for every $f \in  \Dom(\varphi)$.
\item[-] A  term $\Phi \mid \G \ldash M: T$ is computable iff for every computable $\ldash  \varphi : \Phi$ and every $\vdash  \rho:\G$
%
 \[ \den{M}\den{\varphi}\den{\rho}\mydef\;\; \implies \;\;  \varphi \mid \rho \ldash M \myDef\]
 (and similarly for boolean terms).
 
 \end{itemize}
Strictly speaking, in the above we should say that it is the sequent $\Phi \mid \G \ldash M: T$ that is computable and similarly for closures and  function environments. \POPLomit{However we make sure below that any missing information can be supplied from the context.}


\POPLomit{\begin{lemma} \label{helpful}Suppose we have that $\vdash \varphi: \Phi$ and $\Phi[f\type T\rightarrow U] \mid x \type T \vdash M:U$, and 
also that $\varphi$ and $M$ are computable. Then, for all $n$,  the closure 
\[\vdash  \afclo{n}{\varphi}{f(x:T):  U.\, M}: T \rightarrow U\]
 is computable.
\end{lemma}}
\POPLomit{ \begin{proof}
 We prove this by induction on $n$. The result is immediate for $n=0$. For $n +1$,
 we have to show that for all $\vdash V:T$ we have:
 \[\begin{array}{l}\den{M}\den{\varphi[\afclo{n-1}{\varphi}{f(x\type T)\type  U.\, M}/f]}\{ x \mapsto \den{V}\}\mydef \\
   \hspace{50pt} \implies  \varphi[\afclo{n-1}{\varphi}{f(x\type T)\type  U.\, M}/f], \{ x \mapsto  V\} \vdash M \myDef
 \end{array}\]
 But as $\varphi$ and $\afclo{n-1}{\varphi}{f(x\type T):  U.\, M}$ are computable, this follows immediately from the computability of $M$. 
\end{proof}}

%
\begin{lemma} \label{comp}
\hspace{2em}
\begin{enumerate}
\item Every closure $\ldash  \afclo{n}{\varphi}{f(x:T):  U.\, M}: T \rightarrow U$ is computable.
\item Every function environment $\ldash  \varphi:\Phi$ is computable.
\item Every term $\Phi \mid \G \ldash M: T$ is computable.
\item Every  boolean term $\Phi \mid \G \ldash B$ is computable.
\end{enumerate}
\end{lemma}

The next two lemmas enable us to lift completeness from the approximation language to the main one. The first lets us pass from semantic existence in the main language to semantic existence in the approximation language; the second allows us to pass in the opposite direction from termination in the approximation language to termination in the main one.
\myacomment{expand statement when doing long version!}
\begin{lemma} \label{supps}
For any  well-typed term $M$ of the main language we have:
%
%
\[\den{M} = \bigvee_{n \in \N}\den{M^{(n)}}\]
 and similarly for boolean terms, closures, and function environments.
%
\end{lemma}
\POPLomit{
\begin{proof}
The first part is proved by induction on $M$ using the fact that least fixed points are the sup of their approximants.  The other two parts follows immediately.
\end{proof}
}

\begin{lemma} \label{going-up} For any term $M$ of the main language we have:
\[\varphi^{(n)} \!\mid \!\rho \ldash M^{(n)} \imp V \implies \varphi \!\mid\! \rho \vdash M \imp V \]
and
\[\varphi^{(n)} \!\mid \!\rho \ldash M^{(n)} \imps C \implies \varphi \!\mid\! \rho \vdash M \imps C \]
and similarly for boolean terms.
\end{lemma}
\POPLomit{ \begin{proof} Deleting all the $-^{(n)} $s from the proof of $\varphi^{(n)} \mid \rho \vdash M^{(n)} \!\imp\! V$, one  obtains a proof of 
$\varphi \mid \rho \vdash M \!\imp\! V$.
\end{proof}}
%
Operational completeness  follows straightforwardly from these three lemmas:
%
%
 \begin{theorem}[Operational completeness] \label{op-com} 
Suppose  that $\Phi\mid \G \vdash M:T$, $\vdash \varphi\type \Phi$, and  
$\vdash \rho\type \G$. Then:
 \begin{enumerate}
\item \textbf{Operational semantics.}  %
\[\den{M}\den{\varphi}\den{\rho}\mydef \implies \varphi \mid \rho \vdash M\myDef  \]
 (and similarly for boolean terms).
\item \textbf{Symbolic operational semantics.} 
\[\den{M}\den{\varphi}\den{\rho}\mydef \implies \exists C.\, \varphi \mid \rho \vdash M \imps C \]
\end{enumerate}
\end{theorem}
 \begin{proof}
\hspace{1pt}
\begin{itemize}
\item[(1)]  Suppose $\den{M}\den{\varphi}\den{\rho}\mydef$.  
By Lemma~\ref{supps} $\den{M}\den{\varphi} = \bigvee_n \den{M^{(n)}}\den{\varphi^{(n)}}$.
So $\den{M^{(n)}}\den{\varphi^{(n)}}\den{\rho}\mydef$ for some $n$.
 By Lemma~\ref{comp}, both  $M^{(n)}$ and $\varphi^{(n)}$ are computable. 
 So $\varphi^{(n)} \mid \rho \ldash M^{(n)} \myDef$. By Lemma~\ref{going-up}, we then have $\varphi \mid \rho \vdash M \myDef$, as required.
 The proof for boolean terms is similar.
\item[(2)]  This follows from the first part and  Proposition~\ref{ord-symb-con}.
\end{itemize}
\end{proof}

 \section{Discussion}    \label{query}
 

There is much more to do on the theory of differentiable programming languages, even at a basic level; we  briefly suggest some possibilities.
Trace-based automatic differentiation systems generally work with A-normal forms, or equivalent structures. They also employ optimizations. For example, as mentioned in the Introduction, they may record auxiliary information in evaluation traces to reduce recomputation. Other automatic differentiation systems rely on code transformations for differentiation. It would be interesting to define and study 
such optimizations and alternative approaches, perhaps in the setting of our  language.  

Another interesting 
possibility would be to work with non-differentiable functions 
 like ReLU or with non-smooth functions.   For the former, one might use Clarke sub-gradients~\cite{C90}, following~\cite{GE13} (the Clarke sub-gradient of ReLU at $0$ is the interval $[0,1]$); for the latter, one may use $C^k$ functions.
Yet another possibility
  would be to  work with approximate reals, rather than reals, and to seek numerical accuracy theorems; 
  %
  one might employ a   domain-theoretic notion   of sub-differentiation of functions over Scott's interval domain,  generalizing the Clarke sub-gradient (see~\cite{EdalatL04,EdalatM18}).


 One would like results for richer languages,  with a wider range of types or with computational effects. 
 The problem is then how these additional features interact with differentiation. An extension to side-effects would make contact with the literature on the automatic differentiation of imperative languages, such as Fortran. An extension to probability, in some form, would make contact with stochastic optimization and, further, with probabilistic languages for statistical learning.
 For higher-order types, there may be a domain-theoretic analogue of convenient vector spaces that 
 additionally supports recursion. 
 Further, one might, as suggested in~\cite{VKS18}, 
 seek a domain-theoretic analogue of diffeological spaces  (see~\cite{IZ13}); that would also accommodate sum and recursive types. One might also wish to program with Riemannian manifolds,  to accommodate natural gradient descent~\cite{A96}; these too should fit into a diffeological framework.
 In another direction, the work on categories with differential structure may yield an axiomatic version of adequacy theorems for programming languages with differentiation constructs; such categories further equipped with  structure to model partiality~\cite{cockett2011differential} are of particular
interest.

 Finally, if perhaps orthogonally, it is important to add explicit tensor (multi-dimensional array) types, 
with accompanying shape analysis. There is  a long history of programming-language design in this area; a salient example is the design of Remora~\cite{SOM14}.

\bibliography{full_paper_pre_POPL.bib}



\end{document}